\documentclass[10pt]{article}
\usepackage[margin=1in,dvips]{geometry}
\usepackage{graphicx,cancel}
\usepackage{graphics}
\usepackage[]{float}
\usepackage{color, soul}
 \usepackage{xcolor}
\usepackage{amsmath, amsthm, slashed}
\usepackage{tikz-cd} 
\usepackage{pgfplots}
\usetikzlibrary{decorations.pathmorphing}
\usetikzlibrary{decorations.markings}
\usetikzlibrary{patterns}
\usetikzlibrary{shapes}
\usetikzlibrary{plotmarks}
 \usepackage[hidelinks]{hyperref}
\usepackage{amssymb}
\usepackage{mathabx}
\usepackage{rotating}
\usepackage{mathrsfs}
\usepackage{epsfig}
\usepackage{verbatim}
\usepackage[affil-it]{authblk}
\usepackage{rotating}
\usepackage{graphicx}
\usepackage{accents}

\usepackage[numbers,compress]{natbib}
\newtheorem{definition}{Definition}[section]
\newtheorem{theorem}{Theorem}[section]
 
\newtheorem*{conjecture*}{Conjecture}
\newtheorem{corollary}{Corollary}[section]
\newtheorem*{theorem*}{Theorem}
\newtheorem*{corollary*}{Corollary}
\newtheorem{proposition}{Proposition}[section]
\newtheorem{lemma}{Lemma}[subsection]
\newtheorem{remark}{Remark}[section]

\newcommand{\bea}{\begin{eqnarray}}
\newcommand{\eea}{\end{eqnarray}}
\def\beaa{\begin{eqnarray*}}
\def\eeaa{\end{eqnarray*}}

\newcommand{\Xb}{\underline{X}}
\newcommand{\lb}{\underline{\lambda}}
\newcommand{\divs}{\slashed{\mathrm{div}}}
\newcommand{\ns}{\slashed{\nabla}}
\newcommand{\Ds}{\slashed{\mathcal{D}}^{\star}}

\newcommand{\olinb}{\accentset{\scalebox{.6}{\mbox{\tiny (1)}}}{\underline{\omega}}}
\newcommand{\xlin}{\accentset{\scalebox{.6}{\mbox{\tiny (1)}}}{{\hat{\chi}}}}
\newcommand{\xblin}{\accentset{\scalebox{.6}{\mbox{\tiny (1)}}}{\underline{\hat{\chi}}}}
\newcommand{\xflin}{\accentset{\scalebox{.6}{\mbox{\tiny (1)}}}{{\mathfrak{x}}}}
\newcommand{\xfblin}{\accentset{\scalebox{.6}{\mbox{\tiny (1)}}}{{\underline{\mathfrak{x}}}}}

\newcommand{\Olin}{\Omega^{-1}\accentset{\scalebox{.6}{\mbox{\tiny (1)}}}{\Omega}}
\newcommand{\Olino}{\accentset{\scalebox{.6}{\mbox{\tiny (1)}}}{\Omega}}
\newcommand{\glinh}{\accentset{\scalebox{.6}{\mbox{\tiny (1)}}}{\hat{\slashed{g}}}}

  \newcommand{\glinto}{\accentset{\scalebox{.6}{\mbox{\tiny (1)}}}{\sqrt{\slashed{g}}}}
\newcommand{\bmlin}{\accentset{\scalebox{.6}{\mbox{\tiny (1)}}}{b}}

\newcommand{\chiblin}{\accentset{\scalebox{.6}{\mbox{\tiny (1)}}}{\underline{\hat{\chi}}}}
\newcommand{\chilin}{\accentset{\scalebox{.6}{\mbox{\tiny (1)}}}{{\hat{\chi}}}}

\newcommand{\eblin}{\accentset{\scalebox{.6}{\mbox{\tiny (1)}}}{\underline{\eta}}}
\newcommand{\elin}{\accentset{\scalebox{.6}{\mbox{\tiny (1)}}}{{\eta}}}
\newcommand{\otx}{\accentset{\scalebox{.6}{\mbox{\tiny (1)}}}{(\Omega \mathrm{tr} \chi)}}
\newcommand{\otxb}{\accentset{\scalebox{.6}{\mbox{\tiny (1)}}}{(\Omega \mathrm{tr} \underline{\chi})}}
\newcommand{\olin}{\accentset{\scalebox{.6}{\mbox{\tiny (1)}}}{\omega}}
\newcommand{\oblin}{\accentset{\scalebox{.6}{\mbox{\tiny (1)}}}{\underline{\omega}}}

\newcommand{\ablin}{\accentset{\scalebox{.6}{\mbox{\tiny (1)}}}{\underline{\alpha}}}
\newcommand{\aalin}{\accentset{\scalebox{.6}{\mbox{\tiny (1)}}}{\underline{\alpha}}}
\newcommand{\alin}{\accentset{\scalebox{.6}{\mbox{\tiny (1)}}}{{\alpha}}}

\newcommand{\bblin}{\accentset{\scalebox{.6}{\mbox{\tiny (1)}}}{\underline{\beta}}}
\newcommand{\blin}{\accentset{\scalebox{.6}{\mbox{\tiny (1)}}}{{\beta}}}
\newcommand{\rlin}{\accentset{\scalebox{.6}{\mbox{\tiny (1)}}}{\rho}}
\newcommand{\slin}{\accentset{\scalebox{.6}{\mbox{\tiny (1)}}}{{\sigma}}}
\newcommand{\Klin}{\accentset{\scalebox{.6}{\mbox{\tiny (1)}}}{K}}

\newcommand{\bflin}{\accentset{\scalebox{.6}{\mbox{\tiny (1)}}}{{\frak{b}}}}

\def\dual{{\,^\star \mkern-2mu}}
\def\tr{\mbox{tr}}

\newcommand{\nabb}{\nab\mkern-13mu /\,}

\renewcommand{\div}{\mbox{div }}
\newcommand{\curl}{\mbox{curl }}

\def\nab{\nabla}

\def\DDs{ \, \DD \hspace{-2.4pt}\dual    \mkern-16mu /}

\def\a{\alpha}
\def\b{\beta}

\def\om{\omega}

\renewcommand{\aa}{\protect\underline{\a}}
\newcommand{\bb}{\protect\underline{\b}}
\def\omb{{\underline{\om}}}

\newcommand{\chib}{\underline{\chi}}

\def\DD{{\mathcal D}}

\def\ff{\frak{f}}

\def\trch{\tr \chi}
\def\trchb{\tr \chib}

\usepackage{scalerel}

\def\bF{\,^{\scaleto{(F)}{5pt}} \hspace{-2.2pt}\b}
\def\bbF{\,^{\scaleto{(F)}{5pt}} \hspace{-2.2pt}\bb}

\def\rhoF{\,^{\scaleto{(F)}{5pt}} \hspace{-2.2pt}\rho}

\def\sigmaF{\,^{\scaleto{(F)}{5pt}} \hspace{-2.2pt}\sigma}

\def \f12{\frac 1 2 }

\def\rhoFlin{\,^{\scaleto{(F)}{5pt}} \hspace{-1.8pt}\rlin}
\def\sigmaFlin{\,^{\scaleto{(F)}{5pt}} \hspace{-1.8pt}\slin}
\def\bFlin{\,^{\scaleto{(F)}{5pt}} \hspace{-2pt}\blin}
\def\bbFlin{\,^{\scaleto{(F)}{5pt}} \hspace{-2pt}\bblin}

\newcommand{\fflin}{\accentset{\scalebox{.6}{\mbox{\tiny (1)}}}{{\ff}}}

\newcommand{\curls}{\slashed{\mathrm{curl}}}
\newcommand{\ds}{\slashed{\Delta}}
\newcommand{\eps}{\varepsilon}
\DeclareMathAlphabet\mathbfcal{OMS}{cmsy}{b}{n}

\allowdisplaybreaks
\usepackage{upgreek}
\title{Conservation Laws and Boundedness for Linearised Einstein--Maxwell Equations on the Reissner-Nordstr{\"o}m Black Hole}

\author[1]{Marios A. Apetroaie\thanks{marios.apetroaie@epfl.ch}}
\author[2]{Sam C. Collingbourne\thanks{SC.Collingbourne@ed.ac.uk}}
\author[3]{Elena Giorgi\thanks{elena.giorgi@columbia.edu}}

\affil[1]{\small  Institute of Mathematics,
EPFL SB MATH, \vskip.1pc \ MA A2 383 (Bâtiment MA), Station 8, \vskip.1pc \
CH-1015 Lausanne,
Switzerland} 
\affil[2]{\small School of Mathematics, University of Edinburgh, \vskip.1pc \ James Clerk Maxwell Building, Peter Guthrie Tait Road, \vskip.1pc \ Edinburgh, EH9
3FD, Scotland, UK  }
\affil[3]{\small Department of Mathematics, Columbia University, \vskip.1pc \  2990 Broadway,~New York,~NY 10027,~USA  \vskip.2pc \ }

\begin{document}

\maketitle

\begin{abstract}
 We study the linearised Einstein–Maxwell equations on the Reissner–Nordström spacetime and derive the canonical energy conservation law in double null gauge. In the spirit of the work of Holzegel and the second author \cite{ColHol24}, we avoid any use of the hyperbolic nature of the Teukolsky equations and rely solely on the conservation law to establish control of energy fluxes for the gauge-invariant Teukolsky variables, previously identified by the third author~\cite{Giorgi4, Giorgi5}, along all outgoing null hypersurfaces, for charge-to-mass ratio $\frac{|Q|}{M} < \frac{\sqrt{15}}{4}$. This yields uniform boundedness for the Teukolsky variables in Reissner-Nordstr\"om.
\end{abstract}
\tableofcontents

\section{Introduction}

The Einstein–Maxwell equations govern the interaction of gravitational and electromagnetic fields within the framework of General Relativity. These equations describe the geometry of a four-dimensional Lorentzian manifold $({\mathcal{M}}, \mathrm{g})$ coupled to a closed 2-form $\mathrm{F}$ representing the electromagnetic field. The Einstein--Maxwell equations,
\begin{align*}
    \mathrm{Ric}(\mathrm{g})_{ab}&=2\Big(\mathrm{F}_{ac}{\mathrm{F}_{b}}^c-\frac{1}{4}|\mathrm{F}|^2_{\mathrm{g}}\mathrm{g}_{ab}\Big),\qquad
    \mathrm{d}\mathrm{F}=0,\qquad
    \div\mathrm{F}=0,
\end{align*}
admit a $2$-parameter family of solutions known as the Reissner–Nordström spacetimes. These represent static, spherically symmetric black holes endowed with electric charge. They are parametrised by their mass $M$ and charge $Q$, subject to the bound $|Q| \leq M$ to ensure the presence of an event horizon.

In this paper, we revisit the linear stability of the exterior region of the Reissner--Nordstr\"om black hole as solution to the Einstein-Maxwell equations.
We establish a stability statement by \textit{exploiting the existence of a conservation law in the linearised system} rather than utilising `the traditional route' of \textit{studying the Teukolsky equations as hyperbolic equations}. Nevertheless, the method leads to the proof of boundedness for energy fluxes of the Teukolsky variables.

\paragraph{Linear Stability from Conservation Laws.}
The linear stability problem of the Reissner-Nordstr\"om family of black holes involves studying solutions to the linearized Einstein–Maxwell equations around a fixed Reissner–Nordström background. This task is complicated by the gauge freedom inherent in the theory and by the coupled nature of the gravitational and electromagnetic perturbations. A well-known route to overcome these difficulties is to study the perturbations at the level of curvature and derive decoupled equations governing the evolution of gauge-invariant curvature scalars—most notably, the Teukolsky equations. These are then analysed using techniques adapted from the study of wave equations on black hole backgrounds. This framework has proven extremely successful, particularly in recent years, in tackling linear or nonlinear stability questions for Schwarzschild~\cite{DHR,KS1,DHRT}, Kerr~\cite{DHR2,Ma,YakovRita1,YakovRita2,ABBM,KS,GKS}, Reissner-Nordstr\"om~\cite{Giorgi6,Apetroaie24} and Kerr-Newman spacetimes~\cite{Giorgi9}.\footnote{For proofs of linear stability that do not rely on the Teukolsky system see also \cite{Hintz,Lili}.}
However, the use of the Teukolsky formalism, while powerful, comes at a cost: it requires the exploitation of the background’s Petrov type and principal null directions. Moreover, the implementation of standard analytic tools available for the wave equation do not \textit{straightforwardly} carry over to the decoupled Teukolsky equations. Instead one relies on the transformation theory linking the gauge-invariant Teukolsky variables to solutions of better-behaved equations (e.g., Regge–Wheeler and Zerilli-type variables).

In this paper, we pursue a different strategy. We develop a proof of linear stability for the Reissner–Nordström spacetime that avoids altogether the Teukolsky formalism and the associated machinery of decoupled curvature perturbations. Instead, our method exploits the existence of conservation laws that are naturally inherited by the linearized Einstein–Maxwell system from the Killing symmetries of the background spacetime (see already Section~\ref{sec:introduction-canonical} below). By adopting a perspective based on the canonical energy conservation law, and building on the recent approach developed in Schwarzschild by Holzegel and the second author \cite{HolCL16,ColHol24}, we show that these conserved quantities provide uniform control of the gauge-invariant Teukolsky variables. Remarkably, one can establish boundedness of their energy fluxes without referencing their wave-type equations.

A key ingredient in our analysis is the use of the double null gauge, in which the spacetime is foliated by two intersecting families of null hypersurfaces. In double null gauge, one represents the gravitational and electromagnetic perturbations not only through variations of the metric, but also via associated variations in the connection coefficients and curvature components. At the non-linear level, this involves decomposing the metric, connection, and Weyl curvature with respect to a suitably chosen null frame and expressing the geometric structure—via the Levi-Civita conditions, the definition of the Riemann tensor, and the Bianchi identities—in terms of frame components. Linearizing these equations then yields a redundant but highly structured system that is well-suited for analyzing stability problems.
The main advantage of this formulation is the particularly symmetric and simple expression that the conserved quantities take, particularly adept to prove coercivity. 

The approach to linear stability from conserved quantities offers conceptual and technical simplifications. The conservation laws we employ are robust and geometrically natural, providing a direct route to boundedness without relying on transformation theory or intricate algebraic identities. In this way, this method, which was initiated by Holzegel in \cite{HolCL16} in the case of the Schwarzschild black hole, also proves to be fruitful in the broader setting of the Einstein--Maxwell equations by producing a more transparent stability mechanism in Reissner–Nordström.

\paragraph{The System of Teukolsky equations and Previous Works.}

The linearization of the Einstein–Maxwell equations around the Reissner–Nordström background yields a highly coupled system involving perturbations of the metric, the electromagnetic field, and the curvature. Working in a suitably adapted null frame, one finds that the complexity of the system arises not only from gauge freedom but also from the interactions between the gravitational and electromagnetic components. A direct analysis of this full linearized system is, in general, prohibitively difficult.

A more tractable approach—motivated by developments in the vacuum setting—is to isolate gauge-invariant quantities that capture the physical degrees of freedom emitted by the perturbation. In the Schwarzschild and Kerr spacetimes, the celebrated Teukolsky formalism identifies certain null curvature components, denoted $\upalpha^{[\pm2]}$, which satisfy decoupled wave equations. These so-called Teukolsky equations reduce the analysis of the full system to scalar wave-type equations for spin $\pm2$ fields. However, in the Einstein–Maxwell setting, the structure becomes significantly more involved: although gauge-invariant quantities can still be identified, the equations they satisfy no longer fully decouple. Instead, they obey a coupled system of Teukolsky-type equations, involving spin $\pm1$ and $\pm2$ scalar fields.

In prior work by the third author~\cite{Giorgi4, Giorgi5}, a collection of such gauge-invariant quantities was first identified, specifically, $\upalpha$, $\mathfrak{b}$, $\mathfrak{f}$, $\mathfrak{x}$, along with their ingoing counterparts.\footnote{The scalar version equivalent of the quantity $\mathfrak{b}$, corresponding in Newman-Penrose formalism to $2\Psi_1\phi_1-3\phi_0\Psi_2$, appeared in page 240 of \cite{Chandrasekhar2}, where it is noted that such combination is invariant to first order for infinitesimal
rotations.} These tensorial quantities satisfy a system of coupled equations governed by spin-weighted Teukolsky operators, see already equations \eqref{eq:teuk-eq-a}--\eqref{eq:teuk-eq-xx} for their exact expressions. These equations capture the essential features of the linearized dynamics, but the couplings on the right-hand side prevent the application of techniques for decoupled wave equations. Nonetheless, the Teukolsky system for Reissner–Nordström, and its generalization to Kerr–Newman \cite{Giorgi7}, has provided a critical framework for understanding the propagation and decay properties of perturbations.

Despite the absence of conserved energies for the Teukolsky variables at the level of first derivatives---even in vacuum---substantial progress has been made in establishing stability properties of the Reissner–Nordström spacetime. Key breakthroughs were achieved by introducing derived quantities—higher-order combinations of the gauge-invariant fields—that satisfy Regge–Wheeler-type wave equations, first introduced in \cite{DHR} in Schwarzschild. These derived equations allow to derive energy estimates, allowing for the application of the vector field method to obtain uniform energy boundedness, integrated decay, and pointwise control.

In the work \cite{Giorgi7a}, the third author established boundedness and decay estimates for $\fflin, \blin$ and $\underline{\fflin}, {\bblin}$ in the full subextremal range $|Q|<M$ of Reissner–Nordström spacetimes. This analysis was based on the transformation of these quantities into higher-order variables, using the so-called Chandrasekhar transformation. The transformed variables satisfy a symmetric system of coupled Regge–Wheeler-type equations, which admits a combined energy-momentum tensor, allowing for coercive energy and Morawetz estimates in the full subextremal range. Building on this, the first author extended the analysis to the extremal Reissner–Nordström case $|Q|=M$ in \cite{Apetroaie24}, where both stability and instability phenomena along the event horizon, known as Aretakis instabilities, are known to occur.
In these works, the wave equations for $\alin$, $\xflin$ and $\aalin$, $\underline{\xflin}$ have not been used, but rather these quantities have been estimated through simpler transport equations. In the subextremal range $|Q|<M$, control of these gauge-invariant quantities has been proved to estimate all the remaining metric, curvature and electromagnetic components upon a choice of gauge in \cite{Giorgi6}.

A further extension of this framework has been achieved for perturbations of the Kerr–Newman spacetime \cite{Giorgi7}, where a generalisation of the coupled system has been formulated. In this more general, rotating and charged setting, boundedness and decay estimates have been obtained in the regime $|a|, |Q|\ll M$ in \cite{Giorgi9} and for $|a| \ll M, |Q| < M$ in axial symmetry in \cite{GiorgiWan}.  For a proof of linear stability of slowly rotating and weakly charged Kerr-Newman black hole that does not rely on the Teukolsky system see \cite{Lili}.

In the present work, we take a different approach. We provide an alternative proof of the boundedness of the energy fluxes associated with the gauge-invariant quantities along any outgoing null cone. Notably, our method does not rely on the Chandrasekhar transformation or on the use of the Teukolsky equations governing these perturbations. Instead, the argument is grounded in a novel conservation law, for the system of linearised Einstein–Maxwell equations, which arises from the canonical energy.

\subsection{Canonical Energy and Conservation Laws}\label{sec:introduction-canonical}
As mentioned above, the Teukolsky formalism for studying linearised gravity on the Reissner–Nordström background has no known conservation law at the level of first derivatives of the Teukolsky variables. However, if one admits the context of the full system of linearized Einstein–Maxwell equations then indeed one does inherit conservation laws from the symmetries of the Reissner–Nordström background.

A particularly well-known conservation law for stationary backgrounds is the canonical energy which originates in a work of Friedman~\cite{Friedman78} and was developed by Hollands and Wald \cite{HollandsWald13}. Encoded in this work is a systematic method for constructing conserved currents for a linearised Einstein-matter system from a given Einstein--Hilbert action coupled to matter (see for example~\cite{Keir14}).

For the Einstein--Maxwell system, the conserved currents can be constructed by considering antisymmetrized second variations of the Einstein–Hilbert Lagrangian coupled with electromagnetism:
\begin{align}
    L[\mathrm{g},A]\doteq \Big[\mathrm{Scal}[\mathrm{g}]-\frac{\kappa}{4}|\mathrm{F}|^2_{\mathrm{g}}\Big]\eps,
\end{align}
where $\mathrm{F}=dA$ is the usual electromagnetic tensor, $\kappa$ is a coupling constant and $\eps$ is the volume form.
Assume that the background fields $\mathrm{g}$ and $A$ are embedded into a two-parameter family of perturbations $\mathrm{g}=\mathrm{g}(\lambda_1,\lambda_2)$ and $A=A(\lambda_1,\lambda_2)$ and assume that $\mathrm{g}(0)$ satisfies the Einstein–Maxwell equations while the linear variations $h_i=\partial_{\lambda_i}\mathrm{g}(0)$ and $a_i=\partial_iA(0)$ satisfy the linearized system. Under these assumptions, the antisymmetrized second variation of the Lagrangian,
\begin{align*}
    0&=\big[\partial_{\lambda_1}\partial_{\lambda_2} L-\partial_{\lambda_2}\partial_{\lambda_1} L\big]\big|_{\lambda_1=0,\lambda_2=0},
\end{align*}
gives rise to a divergence-free current involving $h_1$ and $h_2$. If one has a stationary background then selecting $(h_1,a_1)=(h,a)$ and $(h_2,a_2)=(L_Th,L_Ta)$ corresponds to the canonical energy current of~\cite{HollandsWald13} (see also~\cite{Col23}).

The obstruction that one runs into when attempting to successfully extract bounds from conservation laws for a linearised Einstein-matter system is the opacity of coercivity (on a given foliation) of the conservation law, even in a well-chosen chosen gauge. This is to be expected given the coupled nature of the Einstein equations.

In~\cite{HolCL16}, Holzegel found a resolution to this coercivity problem in the Schwarzschild case. The conservation laws there were established “by direct inspection” from the system of gravitational perturbations expressed in double null gauge. It was later shown in~\cite{Col23} that Holzegel's conservation laws were in direct correspondence with the canonical energy. The conservation laws in~\cite{HolCL16} are expressed in terms of flux integrals over regions bounded by outgoing and ingoing null hypersurfaces, as illustrated in the Penrose diagram below.
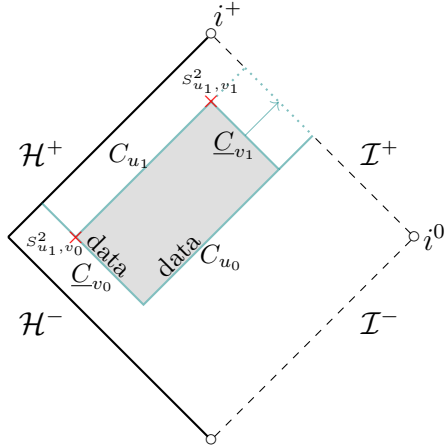
\begin{figure}[H]\label{figure1}
\centering
\begin{tikzpicture}[scale=0.9]
    %SPACETIME
      \draw [thick] (6,3) -- (8.95,5.95);
      \draw [dashed]  (11.95,3.05)--(10.5,4.5);
       \draw [dashed] (9.05,5.95) -- (9.5,5.5);
       \draw [dashed] (9.05,0.05) -- (11.95,2.95);
    \draw [thick] (6,3) -- (8.95,0.05);

      %Char. Rect.
    \draw [thick,teal,fill=lightgray,opacity=0.5] plot  coordinates {(7,3)(9,5)(10,4)(8,2)(7,3)};
    
     \draw [thick,teal,opacity=0.5] (6.5,3.5) -- (7,3);
     \draw [dotted, thick,teal,opacity=0.5] (9,5) -- (9.5,5.5);
      \draw [thick,teal,opacity=0.5] (10,4) -- (10.5,4.5);
       \draw [dotted,thick,teal,opacity=0.5] (9.5,5.5) -- (10.5,4.5);
       
           \draw[teal,->,opacity=0.5] (9.5,4.5) -- (10,5);

      %i+,i0
      
      \node[mark size=2pt,inner sep=2pt] at (12,3) {$\circ$};
      \node[mark size=2pt] at (9,6) {$\circ$};
          \node[mark size=2pt] at (9,0) {$\circ$};

          % Labels     
      \node (scriplus) at (11.5,4.25) {\large $\mathcal{I}^{+}$};
        \node (Hplus) at (6.5,4.25) {\large $\mathcal{H}^{+}$};   
        \node (scriplus) at (11.5,1.75) {\large $\mathcal{I}^{-}$};
        \node (Hplus) at (6.5,1.75) {\large $\mathcal{H}^{-}$};   
      \node (iplus) at (9.25,6.3) {\large $i^{+}$};
      \node (inaught) at (12.3,3) {\large $i^{0}$};
       \node at (7.8,4.2) {$C_{u_1}$};
            \node  at (9.15,2.7) {$C_{u_0}$};
            \node at (7.25,2.4) {$\underline{C}_{v_0}$};
            \node at (9.35,4.35) {$\underline{C}_{{v_1}}$}; 
            \node[rotate=45] at (8.5,2.75) {$\mathrm{data}$};
            \node[rotate=-45] at (7.5,2.75) {$\mathrm{data}$};

            \node at (9,5) {\color{red}{$\times$}};
            \node at (7,3) {\color{red}{$\times$}};
             \node at (9,5.3) {{\tiny $S^2_{u_1,v_1}$}};
            \node at (6.7,2.9) {{\tiny $S^2_{u_1,v_0}$}};

      \end{tikzpicture}  
      \caption{Initial data is prescribed on intersecting outgoing and ingoing null cones $C_{u_0}$ and $\underline{C}_{v_0}$, and the region of interest is bounded by later null cones $C_{u_1}$ and $\underline{C}_{v_1}$. The red crosses denote the spheres $S_{u_1,v_1}^2$ and $S_{u_1,v_0}^2$.}
\end{figure}

In~\cite{HolCL16}, the fluxes on these cones are shown to be gauge invariant up to boundary terms on spheres. By adding an appropriate pure gauge solution normalised to the outgoing cone $C_{u_1}$, the work in~\cite{HolCL16} establishes positivity of the fluxes after sending also the ingoing cone $\underline{C}_{v_1}$ to null infinity. This gives control on the linearised shear of the outgoing null cone $C_{u_1}$ and on the flux radiated to null infinity. A simplifying feature of double null gauge, which aids in establishing coercivity, is that one can readily use the linearised constraint equations.

In~\cite{ColHol24}, this work was extended to show how one can ascend a hierarchy in the system of linearised gravity to control on the gauge-invariant Teukolsky quantities from the coercive estimates obtained from these conservation laws rather than the transformation theory to Regge--Wheeler quantities~\cite{DHR}.

The present paper extends~\cite{HolCL16} and~\cite{ColHol24} to the charged case, namely the non-vacuum Reissner–Nordström spacetime. In particular, as in~\cite{ColHol24}, we demonstrate that without appealing to the full structure of the Teukolsky or Regge–Wheeler systems, one can obtain uniform boundedness for the Teukolsky variables $\bflin$, $\fflin$ and $\alin$, along outgoing null cones. Whilst our result is restricted to the range of charge parameter $|Q|<\frac{\sqrt{15}}{4}M$, our approach simplifies the argument from~\cite{HolCL16} for coercivity of the canonical energy conservation law by \textit{avoiding the use of gauge adjustments along the outgoing null cone}, which becomes particularly simple in the case of Schwarzschild.

\subsection{The Uniform Boundedness Theorem}

Our main result establishes a uniform energy boundedness
for the gauge invariant quantities $\bflin$, $\fflin$ and $\alin$, and their negative spin counterparts, which capture electromagnetic and gravitational perturbations respectively. As mentioned above, the proof relies only on conservation laws and transport estimates in double null gauge — avoiding any use of the Teukolsky or Regge–Wheeler equations. 
Despite the minimal toolkit, we obtain a uniform energy bound along outgoing null hypersurfaces, provided the charge-to-mass ratio satisfies $\frac{|Q|}{M}<\frac{\sqrt{15}}{4}$.

We now state a rough version of the main result:

\begin{theorem}\label{main-theorem} Let $\bflin$, $\fflin$, $\alin$ and  $\underline{\bflin}$, $\underline{\fflin}$, $\ablin$ denote the gauge-invariant Teukolsky variables (defined through equations~\eqref{eq:curvature},~\eqref{eq:definitions-ff} and~\eqref{eq:definitions-bb}). Given a smooth solution of the system of linearised perturbations arising from characteristic initial data on $C_{u_0} \cup\underline{C}_{v_0}$ on the Reissner--Nordstr\"om exterior, with charge-to-mass satisfying $\frac{|Q|}{M}<\frac{\sqrt{15}}{4}$, the following energy estimate holds:
\begin{align}\label{eq:final-statement-theorem}
   \mathrm{E}_{u}[\bflin]+ \mathrm{E}_{u}[\fflin]+\mathrm{E}_{u}[\alin]+ \mathrm{E}_{u}[\underline{\bflin}]+ \mathrm{E}_{u}[\underline{\fflin}]+\mathrm{E}_{u}[\aalin] \lesssim E_{data}(u),
\end{align}
on any outgoing null cone $C_u$, where the $\mathrm{E}_{u}$ denotes standard $L^2$-based coercive energy fluxes involving tangential derivatives of the Teukolsky variables along the outgoing null cone (see \eqref{eq:definition-norms-positive-spin}-\eqref{eq:definition-norms-negative-spin} for the exact definition) and $E_{data}(u)$ denotes a suitable energy norm of the initial data which is uniformly controlled. 

A suitable higher derivative version of the statement also holds upon commutation with the symmetries of the spacetime.
\end{theorem}

This yields uniform $L^2$ control on spheres by standard arguments (using Cauchy–Schwarz and the fundamental theorem of calculus), and further implies pointwise bounds via angular commutation and Sobolev embedding: 
\begin{align*}
\sup|r^4\Omega\bflin|,\qquad \sup|r^2\Omega\fflin|, \qquad \sup |r\Omega^2\alin|, \qquad \sup|r\Omega^2\underline{\fflin}|, \qquad \sup |r^3\Omega^2\underline{\bflin}|, \qquad \sup |r\Omega^3 \aalin| \lesssim \text{initial energy.}
\end{align*}

\subsubsection{Key Steps in the Proof}
The proof proceeds in two main steps:

\paragraph{Step 1: A Coercive Energy Estimate from the Canonical Energy Conservation Law.}

The linearised Einstein–Maxwell system admits a conservation law on the characteristic rectangle shown in Figure 1, see Proposition~\ref{prop:conservationlaws}. Schematically, the conservation law is given by
    \begin{align*}
E_{u_1}[\mathcal{S}](v_0,v_1)+E_{v_1}[\mathcal{S}](u_0,u_1)&=E_{u_0}[\mathcal{S}](v_0,v_1)+E_{v_0}[\mathcal{S}](u_0,u_1).
\end{align*}
where $E_{u_i}[\mathcal{S}](v_0,v_1)$ denotes the flux on the incoming cone $C_{u_i}$ and $E_{v_i}[\mathcal{S}](u_0,u_1)$ denotes the flux on the outgoing cone $\underline{C}_{v_i}$. These energy fluxes are given in terms of Christoffel symbols, electromagnetic and curvature components, see \eqref{eq:def-Eu}-\eqref{eq:def-Ev} for their definition.

We emphasize that the mere existence of a conserved energy does not imply that this energy is positive or coercive, and therefore useful to obtain boundedness statements. Indeed, for general linearized Einstein–Maxwell perturbations, there is no a priori reason to expect positivity, particularly in the presence of coupling between gravitational and electromagnetic radiations. 
The expressions of the energy fluxes $E_{u_i}$ and $E_{v_i}$ are not \textit{manifestly} coercive due to the presence of mixed term coupling curvature and connection components (see the last term in equations~\eqref{eq:def-Eu}-\eqref{eq:def-Ev}). To overcome this, we rewrite the outgoing energy flux as a sum of (partially) gauge-invariant positive terms and two (indefinite sign) boundary terms on the final future sphere $S^2_{u_1,v_1}$ and the initial data sphere $S^2_{u_1,v_0}$ (denoted with red crosses in Figure~\ref{figure1}). Then, to extract a coercive estimate one establishes the following:
\begin{itemize}
    \item The limit of the incoming flux on $\underline{C}_{v_1}$ is coercive as $v_1 \to \infty$ (towards null infinity).
    \item The boundary term final future sphere $S^2_{u_1,v_1}$ can be absorbed into the flux via a transport estimate. It is in this estimate where the $|Q|<\frac{\sqrt{15}}{4}M$ charge restriction comes into the analysis.\footnote{The case $|Q|=\frac{\sqrt{15}}{4}M$ can be dealt with at this stage. However, one loses control of the flux of some quantities which prohibits implementing step 2.} In the case of Schwarzschild, one does not require an estimate at all; the boundary term can be shown to be manifestly positive.
\end{itemize}

Thus, coercivity of the conserved energy does hold in the Reissner–Nordström spacetime under the charge restriction, yielding uniform control on the coercive energy along $C_{u_1}$. This is obtained in Theorem \ref{thm:mastercontrol}.

\paragraph{Step 2: A Hierarchy of Transport Estimates.}

Once control is obtained for certain curvature and connection components from the conservation law, we ascend a hierarchy of transport equations to recover control on the Teukolsky variables $\bflin$, $\underline{\bflin}$, $\fflin$, $\underline{\fflin}$, $\alin$, $\ablin$ and their first derivatives. The derivation of these estimates rely on transport and elliptic estimates, allowing us to propagate them without any further restriction on $Q$. This is obtained in Section \ref{section-hierarchy}.

The combination of the coercive energy of Step 1 and transport estimates here yields the full energy estimate of Theorem~\ref{main-theorem}.

\begin{remark}
Theorem \ref{main-theorem} holds under the suboptimal charge restriction $|Q|<\frac{\sqrt{15}}{4}M$, which is needed in the Step 1 explained above to obtain a coercive estimate from the conservation law (see the proof of Theorem~\ref{thm:mastercontrol}).

While previous results (e.g., in \cite{Giorgi7a}) apply throughout the full subextremal range $|Q|<M$ by analyzing the Teukolsky and Regge–Wheeler systems, the simplicity and robustness of our method — which avoids any use of decoupling or second-order system of wave equations — comes at the cost of this more restrictive condition. It is not clear whether coercivity of the canonical flux can be extended to all subextremal charges, and the threshold $|Q|=\frac{\sqrt{15}}{4}M$ may signal a genuine transition, as extra gauge invariant quantities appear at this critical charge. Moreover, as mentioned above, our method refines the proof of coercivity by expressing the conserved fluxes in terms of gauge invariant quantities, up to boundary terms, and therefore is only affected by gauge choices in a limited manner. For more details on the optimality of the charge-to-mass ratio see Appendix \ref{app:optimality}.
\end{remark}

\subsubsection*{Acknowledgments} M.A.A. acknowledges support by the Alexander von Humboldt Foundation in the framework of the Alexander von Humboldt Professorship of Gustav Holzegel endowed by the Federal Ministry of Education and
Research, as well as Germany’s Excellence Strategy EXC 2044 390685587, Mathematics Münster: Dynamics-Geometry-Structure.  S.C.C acknowledges support from the Royal Society University Research Fellowship URF\textbackslash R1\textbackslash 211216. E.G. acknowledges the support of NSF Grants DMS-2306143, DMS-2336118  and of a grant of the Sloan Foundation.

\section{Preliminaries}

Here we provide preliminaries regarding the Reissner-Nordstr\"om spacetime, the linearised Einstein-Maxwell equations and the conservation laws in this setting.

\subsection{The Reissner--Nordstr\"om Spacetime}

In this section, we summarize the geometric structures of the Reissner–Nordström spacetime that will be used throughout the paper. Our treatment closely follows \cite{DHR, Giorgi6}, to which we refer the reader for further details and background.

The Reissner–Nordström family describes static, spherically symmetric solutions to the Einstein–Maxwell equations, representing charged, non-rotating black holes. 
The metric and electromagnetic tensor are explicitly given by:
\begin{align*}
    g_{M, Q}=-D(r)(dt)^2 +\frac{1}{D(r)}(dr)^2+r^2\gamma_{S^2},\qquad \mathrm{F}=\frac{Q}{r^2}dt\wedge dr,\qquad D(r)\doteq 1-\frac{2M}{r}+\frac{Q^2}{r^2},
\end{align*}
where $\gamma_{S^2}$ is the standard metric on the unit 2-sphere. The coordinates $(t,r,\theta,\phi)$ cover the exterior region $r>r_+$, where $r_+\doteq M+\sqrt{M^2-Q^2}$ the largest root of $r^2D(r)$ and corresponds to event horizon.

We restrict attention to the exterior region of the black hole, where the metric is expressed conveniently in double null Eddington–Finkelstein coordinates
$(u,v,\theta,\phi)\in \mathbb{R}^2\times{S}^2_{u,v}$. These are related to the standard Boyer–Lindquist coordinates $(t,r,\theta,\phi)$ by
\begin{align*}
u=\frac{1}{2}(t-r_{\star}),\qquad v=\frac{1}{2}(t+r_{\star}) , 
\end{align*}
where the Regge–Wheeler (or tortoise) coordinate $r_{\star}(r)$ satisfies
\begin{align*}
\frac{dr_{\star}(r)}{dr}=\frac{1}{1-\frac{2M}{r}+\frac{Q^2}{r^2}} \, .
\end{align*}
 Here, $M$ and $Q$ denote the mass and charge of the black hole, respectively, with the subextremal and extremal range given by $|Q|\leq M$. When $Q=0$, the Reissner–Nordström spacetime reduces to the Schwarzschild solution.

In $(u,v,\theta,\phi)$ coordinates, the Reissner--Nordstr\"om metric  takes the form
\begin{align*}
g=-2\Omega^2(u,v)(du\otimes dv+dv\otimes du)+\slashed{g},\qquad \slashed{g}=r^2(u,v){\gamma},
\end{align*}
where $\Omega^2(u,v)=1-\frac{2M}{r(u,v)}+\frac{Q^2}{r(u,v)^2}$ and~${{\gamma}}$ is the standard metric on the unit~$2$-sphere. 

The level sets of $u$ and $v$, denoted $C_u$ and $\underline{C}_{v}$ respectively, are outgoing and ingoing null hypersurfaces. It is important to note that the double null coordinates $(u,v)$ do not extend to cover the 
future event horizon, $\mathcal{H}^+$, or future null infinity, $\mathcal{I}^+$, but these can still be parametrized formally as $(\infty,v,\theta,\phi)$ and $(u,\infty,\theta,\phi)$.

To analyze tensorial quantities, we introduce a double null frame adapted to this coordinate system:
\begin{align}\label{eq:null-frame}
e_3\doteq\frac{1}{\Omega}\partial_u,\quad e_4\doteq\frac{1}{\Omega}\partial_v
\end{align}
together with a local orthonormal frame $\{e_A\}_{A=1,2}$ tangent to the spheres $S^2_{u,v}$. We primarily consider tensors tangent to $S^2_{u,v}$, which vanish upon contraction with $e_3$ or $e_4$.

Following \cite{CK93, DHR, Giorgi6}, we define the Ricci coefficients:
\begin{align*}
    \chi_{AB}\doteq g(\nabla_{e_A}e_4, e_B), \qquad \underline{\chi}_{AB}\doteq g(\nabla_{e_A}e_3, e_B), \\
    \eta_A\doteq -\frac 1 2 g(\nabla_{e_3}e_A, e_4), \qquad \underline{\eta}_A\doteq -\frac 1 2 g(\nabla_{e_4}e_A, e_3), \\
   \omega\doteq-\frac{\Omega}{2}g(\nabla_{e_4}e_4,e_3) , \qquad \underline{\omega}\doteq-\frac{\Omega}{2}g(\nabla_{e_3}e_3,e_4),
\end{align*}
as well as the null Weyl curvature components:
\begin{equation}
\begin{aligned}\label{eq:curvature}
    \alpha_{AB}\doteq\mathrm{W}(e_A, e_4, e_B, e_4), \qquad \underline{\alpha}_{AB}\doteq\mathrm{W}(e_A, e_3, e_B, e_3), \\
    \beta_A\doteq\frac 1 2 \mathrm{W}(e_A, e_4, e_3, e_4), \qquad   \underline{\beta}_A\doteq\frac 1 2 \mathrm{W}(e_A, e_3, e_3, e_4),\\
    \rho\doteq\frac{1}{4}\mathrm{W}(e_3,e_4,e_3,e_4), \quad \sigma\doteq\frac{1}{4}[\star\mathrm{W}](e_3,e_4,e_3,e_4),
\end{aligned}
\end{equation}
and the null electromagnetic components:
\begin{align*}
    \bF\doteq \mathrm{F}(e_A, e_4), \qquad \bbF\doteq \mathrm{F}(e_A, e_3), \\
    \rhoF\doteq \frac{1}{2}\mathrm{F}(e_3,e_4), \qquad \sigmaF\doteq \frac 1 2 [\star \mathrm{F}](e_3,e_4).
\end{align*}
We also denote $\hat{\om}=\Omega^{-1}\om$ and $\hat{\omb}=\Omega^{-1}\omb$ for convenience. 

In the frame \eqref{eq:null-frame}, the only non-vanishing Ricci coefficients are:
\begin{align*}
\Omega\mathrm{tr}\chi\doteq\Omega \slashed{g}^{AB}\chi_{AB}=\frac{2\Omega^2}{r},\qquad  \Omega\mathrm{tr}\underline{\chi}\doteq\Omega \slashed{g}^{AB}\underline{\chi}_{AB}=-\frac{2\Omega^2}{r}, \qquad 
\omega=\frac{M}{r^2}-\frac{Q^2}{r^3},\qquad  \underline{\omega}=-\frac{M}{r^2}+\frac{Q^2}{r^3},
\end{align*}
and the only non-zero null curvature and electromagnetic components are
\begin{align*}
\rho=-\frac{2M}{r^3}+\frac{2Q^2}{r^4}, \qquad \rhoF=\frac{Q}{r^2}.
\end{align*}
The Gaussian curvature $K$ of each sphere $(S^2_{u,v}, \slashed{g})$ satisfies
\begin{align*}
    K=\frac{1}{r^2}.
\end{align*}

Several important relations in double null gauge will be frequently used: 
\begin{align*} 
\ns_3\Omega=\omb, \qquad \ns_4\Omega=\om,\qquad \partial_v r= \Omega^2, \qquad \partial_u r = -\Omega^2.
\end{align*} 
The spacetime possesses symmetries encoded by a static Killing vector field $T$, given in double null coordinates by
\begin{align}\label{eq:definition-T}
T=\frac{1}{2}(\partial_u+\partial_v)=\frac{\Omega}{2}(e_3+e_4).
\end{align}

Finally, associated to the spherical symmetry of the Schwarzschild spacetime, we have three angular momentum operators generating the Lie algebra $\mathfrak{so}(3)$, which can be expressed explicitly expressed in standard $(\theta,\phi)$ coordinates by
\begin{align*}
\Omega_1=\sin\phi\partial_{\theta}+\cot\theta\cos\phi\partial_{\phi},\qquad\Omega_2=\cos\phi\partial_{\theta}-\cot\theta\sin\phi\partial_{\phi},\qquad \Omega_3=\partial_{\phi}.
\end{align*}

\subsection{The $S^2_{u,v}$-Tensor Algebra}

We begin by recalling the main conventions and norms associated with the algebra of $S^2_{u,v}$-tensor fields. Throughout, let $\Theta$ and $\Phi$ denote $n$-covariant $S^2_{u,v}$-tensors. We define the pointwise inner product and associated norms via
\begin{align*}
\langle \Theta,\Phi\rangle&\doteq\slashed{g}^{A_1B_1}...\slashed{g}^{A_nB_n}\Theta_{A_1...A_n}\Phi_{B_1...B_n},\qquad
|\Theta|^2\doteq\langle\Theta,\Theta\rangle, 
\end{align*}
where $\slashed{g}$ denotes the induced metric on $S^2_{u,v}$. Integrating over the spheres, we define the $L^2$ inner product and corresponding norm by
\begin{align*}
\langle \Theta,\Phi\rangle_{u,v}&\doteq\int_{S^2}\langle \Theta,\Phi\rangle(u,v,\theta,\phi)\varepsilon_{S^2}, \qquad ||\Theta||_{u,v}\doteq \left(\int_{S^2}|\Theta|^2(u,v,\theta,\phi)\varepsilon_{S^2}\right)^{\frac 12},
\end{align*}
where $\varepsilon_{S^2}$ denotes the induced volume form on $S^2$.

To handle derivatives intrinsic to the spheres, we introduce the projected covariant derivatives $\ns_3$, $\ns_4$, and $\ns_A$. These are defined by extending $\Theta$ trivially to a spacetime covariant tensor field and subsequently projecting the spacetime covariant derivatives onto the sphere frame ${e_A}$. A direct computation yields
\begin{align*}
\ns_3\Theta_{A_1...A_n}&=e_3(\Theta_{A_1...A_n})-\frac{1}{2}\mathrm{tr}\underline{\chi}\Theta_{A_1...A_n} \, , \qquad \ns_4\Theta_{A_1...A_n}=e_4(\Theta_{A_1...A_n})-\frac{1}{2}\mathrm{tr}{\chi}\Theta_{A_1...A_n} \,.
\end{align*}
We also define a projected derivative along the vector field $T$, consistent with \eqref{eq:definition-T}, as
\begin{align*}
\ns_T\Theta&\doteq\frac{\Omega}{2}\big(\ns_3\Theta+\ns_4\Theta\big).
\end{align*}
In a similar spirit, we define the projected Lie derivatives $\slashed{\mathcal{L}}_T$ and $\slashed{\mathcal{L}}_{\Omega_k}$, corresponding to the vector fields $T$ and $\Omega_k$, respectively, by
\begin{align*}
(\slashed{\mathcal{L}}_T\Theta)_{A_1...A_n}&\doteq(\mathcal{L}_T\Theta)_{A_1...A_n}\qquad(\slashed{\mathcal{L}}_{\Omega_k}\Theta)_{A_1...A_n}\doteq({\mathcal{L}}_{\Omega_k}\Theta)_{A_1...A_n}
\end{align*}
Since $T$ is a linear combination of $e_3$ and $e_4$, it follows that
\begin{align*}
(\slashed{\mathcal{L}}_T\Theta)_{A_1...A_n}=(\ns_T\Theta)_{A_1...A_n}=T(\Theta_{A_1...A_n}),
\end{align*}
where $T(\Theta)$ denotes the standard directional derivative along $T$.

For divergence and curl operations on $S^2_{u,v}$, we define the $(n-1)$-covariant tensor fields $\slashed{\mathrm{div}}\Theta$ and $\slashed{\mathrm{curl}}\Theta$ by
\begin{align*}
(\slashed{\mathrm{div}}\Theta)_{A_1...A_{n-1}}&\doteq\slashed{g}^{BC}(\slashed{\nabla}_B\Theta)_{CA_1...A_{n-1}},\qquad
(\slashed{\mathrm{curl}}\Theta)_{A_1...A_{n-1}}\doteq\slashed{\varepsilon}^{BC}(\slashed{\nabla}_B\Theta)_{CA_1...A_{n-1}},
\end{align*}
where $\slashed{\varepsilon}$ is the induced volume form on ${S}^2_{u,v}$. 

Several useful identities govern these angular operators, as discussed in \cite{ColHol24}. In particular, for functions, $1$-forms, and symmetric traceless $2$-tensors $\Theta_1,\Theta_2$, one has
\begin{align}
\sum_{k=1}^3 \left\langle \slashed{\mathcal{L}}_{\Omega_k} \Theta_1, \slashed{\mathcal{L}}_{\Omega_k} \Theta_2 \right\rangle =
\begin{cases}
\left\langle r\ns \Theta_1, r\ns \Theta_2 \right\rangle, & \text{functions} \\
\left\langle r\ns \Theta_1, r\ns \Theta_2 \right\rangle + \left\langle \Theta_1, \Theta_2 \right\rangle, & \text{1-forms} \\
\left\langle r\ns \Theta_1, r\ns \Theta_2 \right\rangle + 4 \left\langle \Theta_1, \Theta_2 \right\rangle, & \text{symm-traceless 2-tensors}
\end{cases}\label{eq:LieComId}
\end{align}
Moreover, the angular momentum vector fields satisfy the algebraic identities
\begin{align*}
\sum_k(\slashed{\mathrm{curl}}\Omega_k)^2=4,\qquad \sum_k(\slashed{\mathrm{curl}}\Omega_k)\Omega_k=0,\qquad \sum_{k}\Omega_k^A\Omega_k^B&=r^2\slashed{g}^{AB}.
\end{align*}

We now introduce further differential operators on $S^2_{u,v}$. For a one-form $\xi$, we define the symmetric trace-free angular derivative operator
\begin{align*}
\slashed{\mathcal{D}}_2^{\star}\xi&\doteq-\frac{1}{2}\Big[(\slashed{\nabla}_A\xi)_B+(\slashed{\nabla}_B\xi)_A-(\slashed{\mathrm{div}}\xi)\slashed{g}_{AB}\Big] \, ,
\end{align*}
which is the formal $L^2$-adjoint of $\slashed{\mathrm{div}}$. Similarly, for a pair of scalar functions $(f_1, f_2)$ on $S^2_{u,v}$, we define
\begin{align*}
    [\slashed{\mathcal{D}}_1^{\star}(f_1,f_2)]_A&\doteq-\slashed{\nabla}_Af_1+\slashed{\varepsilon}_{AB}\slashed{\nabla}^Bf_2,
\end{align*}
which is the formal $L^2$-adjoint of the map $\xi \mapsto (\slashed{\mathrm{div}}\xi, \slashed{\mathrm{curl}}\xi)$.

The Laplacian $\ds$ acting on tensor fields over $S^2_{u,v}$ is defined as
\begin{align*}
\ds\Theta\doteq\divs(\ns\Theta),
\end{align*}
and satisfies the scaling relation $\Delta_{S^2} = r^2 \ds$, where $\Delta_{S^2}$ denotes the standard Laplacian on the unit sphere.

The following fundamental identity for $S^2_{u,v}$-vector fields $X$ will be important.
\begin{lemma}\label{lem:nabslashedtodivcurlest} Let $X$ be an $S^2_{u,v}$ one-form. Then we have the identity
\begin{align*}
    \int_{S^2}{|r\ns X|^2}\eps_{S^2}=\int_{S^2}\Big[|r\divs X|^2+|r\curls X|^2-|X|^2\Big]\eps_{S^2}.
    \end{align*}
\end{lemma}
\begin{proof}
This follows from a direct computation, exploiting the relation $\slashed{\varepsilon}_{AB}\slashed{\varepsilon}_{CD} = \slashed{g}_{AC}\slashed{g}_{BD} - \slashed{g}_{AD}\slashed{g}_{BC}$ and applying the Ricci identity.
\end{proof}

\subsubsection{Support on $\ell\geq 1$ Spherical Harmonics}

In this work, we consider functions, one-forms and symmetric-traceless $2$-tensors on $S^2_{u,v}$  that have support on $\ell\geq 1$. 
A function $f$ is said to be supported on $\ell\geq 1$ if its projection onto the $\ell=0$ spherical harmonics vanishes. We note that one-forms on $S^2_{u,v}$ are automatically supported on $\ell \geq 1$, while symmetric-traceless $2$-tensors on $S^2_{u,v}$ are automatically supported on $\ell \geq 2$.

For any function $f$ on $S^2_{u,v}$ supported on $\ell\geq 1$, or for any one-form $f$ on $S^2_{u,v}$, the Poincar\'e inequality holds:
\begin{align*}
 \|f\|_{u,v}\leq \|r\ns f\|_{u,v}.
\end{align*}

\subsection{Transport Lemma Estimate}

We conclude with a crucial energy-type estimate for tensors transported along the outgoing null cones.

\begin{lemma} \label{lem:tplemma}
Let $\xi$ and $\Xi$ be $S^2_{u,v}$ tensor fields satisfying
\begin{align*}
    \Omega \slashed{\nabla}_4 \xi = \Omega^2 \Xi
\end{align*}
along a cone $C_{u}$ for $u_0 \leq u \leq \infty$, and assume $\xi = \mathcal{O}(\Omega)$ if $Q<M$
as $u\rightarrow \infty$. Then,
        \begin{align*}  \|\xi\|^2(u,\infty) & \leq  \frac{1}{1-\frac{r_+}{r(u,v_0)}}\|\xi\|^2(u,v_0)+\frac{1}{r_+}\int_{v_0}^{\infty}\big|\Omega r\Xi\big|^2dv\eps_{S^2},
    	\end{align*}
        and
        \begin{align*}   \int_{v_0}^{\infty}\frac{\Omega^2}{r^2}|\xi|^2dv\eps_{S^2} & \lesssim \|\xi\|^2(u,v_0)+\frac{1}{r_+}\int_{v_0}^{\infty}\big|\Omega r\Xi\big|^2dv\eps_{S^2},
    	\end{align*}
        where we denote $f \lesssim g$ if there exists a uniform constant $C$ such that $f \leq Cg$.
\end{lemma}

\begin{proof}
The proof proceeds by introducing an auxiliary function $F=F(r)$ satisfying a differential inequality and optimizing the choice of $F$ to control the growth of $|\xi|^2$ along the cone. In fact, for any function $F=F(r)$, we compute
    	\begin{align*}
    	    \partial_v \left(F |\xi|^2\right) & =  \Omega^2F'|\xi|^2  + 2 \Big(\frac{\sqrt{\eps}\Omega r\Xi }{\sqrt{r_+}}\Big)\cdot \Big(\frac{\sqrt{r_+}F\Omega\xi}{r\sqrt{\eps}}\Big)\\
            &=\Omega^2\Big[F'+\frac{r_+F^2}{r^2\eps}\Big]|\xi|^2  + \Big|\frac{\sqrt{\eps}\Omega r\Xi }{\sqrt{r_+}}\Big|^2-\Big|\frac{\sqrt{\eps}\Omega r\Xi }{\sqrt{r_+}}-\frac{\sqrt{r_+}F\Omega\xi}{r\sqrt{\eps}}\Big|^2,
    	\end{align*}
      where we use $ \partial_vF = \Omega^2 F' $. The idea is to find $F$ such that, for the smallest possible $\epsilon>0$, the prefactor of $|\xi|^2$ on the right-hand side is non-positive; it suffices to be zero, i.e. 
      \begin{align*}
     F'+\frac{r_+}{\eps r^2}F^2=0.
\end{align*}
For $\delta\in \mathbb{R}$, we look for solutions to the ODE
\begin{align*}
     F'+\frac{\delta r_+}{ r^2}F^2=0.
\end{align*}
Requiring the boundary condition $F(r=\infty)=1$, we have the following explicit solution of the ODE 
\begin{align*}
   F(r)=
   \frac{1}{1-\frac{\delta r_+}{r}}
\end{align*}
We obtain that
\begin{align*}
    	   |\xi|^2(u,\infty)+\int_{v_0}^{\infty}\Big[\Big|\frac{\sqrt{\eps}\Omega r\Xi }{\sqrt{r_+}}-\frac{\sqrt{r_+}F\Omega\xi}{r\sqrt{\eps}}\Big|^2+(\delta\eps-1)\Big|\frac{\sqrt{r_+}F\Omega \xi}{ r\sqrt{\eps}}\Big|^2\Big]dv & = \frac{1}{1-\frac{\delta r_+}{r}}|\xi|^2(u,v_0)+\frac{\eps}{r_+}\int_{v_0}^{\infty}\big|\Omega r\Xi\big|^2dv,
    	\end{align*}
    One can then integrate on the spheres to obtain the estimate,
    \begin{align*}  \|\xi\|^2(u,\infty)+\int_{v_0}^{\infty}\frac{r_+(\delta\eps-1)}{\eps}\Big|\frac{\Omega \xi}{r-\delta r_+}\Big|^2dv\eps_{S^2} & \leq \frac{1}{1-\frac{\delta r_+}{r(u,v_0)}}\|\xi\|^2(u,v_0)+\frac{\eps}{r_+}\int_{v_0}^{\infty}\big|\Omega r\Xi\big|^2dv\eps_{S^2},
    	\end{align*}
    modulo understanding the restriction on $(\delta,\eps)$. This comes from smoothness of $F|\xi|^2$. We can guarantee smoothness of $F|\xi|^2$ for any $r(u,v_0)> r_+$ by picking $\delta\leq 1$. For $\delta=1$, notice that in view of $\Omega^2 = \frac{1}{r^2}(r-r_+)(r-r_{-})$ we can re-express $ F = \Omega^{-2}\left(1-\frac{r_{-}}{r}\right)$, so if $\xi = \mathcal{O}(\Omega)$, if $Q<M$
    as $u\to\infty$ 
 along the initial cone $\underline{C}_{v_0}$ then $F\cdot |\xi|^2$ extends smoothly on the horizon. Picking $\delta=1$ requires $\eps\geq 1$. The equality case gives the first estimate. Picking $\delta>0$ sufficiently small and $\eps$ sufficiently large gives the second estimate. 
    \end{proof}

\subsection{The Linearised Einstein--Maxwell Equations in Double Null Gauge}
The system of linearised Einstein--Maxwell equations in double null gauge is encoded by the linearised metric quantities
\begin{align*}
\glinto, \glinh_{AB}, \bmlin_A,\Olin,
\end{align*}
the linearised connection coefficients
\begin{align*}
\otx,\otxb,\olin,\olinb,\elin_A,\eblin_A, \xlin_{AB},\xblin_{AB} ,
\end{align*}
the linearised curvature components
\begin{align*}
\rlin,\slin,\blin_A,\bblin_A, \alin_{AB},\ablin_{AB},
\end{align*}
and the linearised electromagnetic components
\begin{align*}
   \rhoFlin, \sigmaFlin,  \bFlin_A, \bbFlin_A.
\end{align*}
The above quantities can be interpreted as the linear perturbations of the respective quantities with respect to the Reissner-Nordstr\"om background value.

We will denote $\mathcal{S}$ a solution to the system of linearised electro-gravitational perturbations intended as a collection of quantities 
\begin{align*} 
\mathcal{S}=\bigg(\, \glinh \, , \, \glinto \, , \, \Olino \, , \,  \bmlin\, , \,  \otx \, , \,  \otxb\, , \,  \xlin\, , \, \xblin\, , \,  \eblin \, , \,  \elin \, , \, \olin \, , \,  \olinb \, , \,  \alin \, , \,  \blin \, , \,  \rlin \, , \,  \slin \, , \,  \bblin \, , \,  \ablin \, , \, \Klin\, , \, \bFlin\, , \,\rhoFlin\, , \, \sigmaFlin \, , \, \bbFlin \bigg)
\end{align*}
satisfying the system \eqref{eq:first}--\eqref{nabb-4-check-sigma} below, which we call the system of linearised Einstein--Maxwell equations on the Reissner--Nordstr\"om background.  It comprises of the following system of equations.

\paragraph{Linearised Null Structure Equations:} 
    \begin{align}
\partial_u (\frac{\glinto}{\sqrt{\slashed{g}}})&= \otxb, \qquad
\partial_v (\frac{\glinto}{\sqrt{\slashed{g}}})=\otx - \div \bmlin, \label{eq:first}\\
\sqrt{\slashed{g}} \partial_u (\frac{\glinh}{\sqrt{\slashed{g}}})&=  2\Omega \chiblin, \qquad
\sqrt{\slashed{g}} \partial_v (\frac{\glinh}{\sqrt{\slashed{g}}})= 2\Omega \chilin +2\DDs_2 \bmlin,\\
\partial_u\bmlin&=2\Omega^2(\elin-\eblin) \\
\partial_v(\Olin)&=\olin,\qquad
\partial_u(\Olin)= \oblin ,\qquad 2\nabb(\Olin)=\elin+\eblin,\label{eq:pr-Olin}\\
\partial_v\otxb &= \Omega^2 (2\divs  \eblin + 2 \rlin + 4 \rho \Olin) -\frac 1 2 \Omega \trch (\otxb-\otx)\label{D-otxb}\\
\partial_u\otx &= \Omega^2 (2\divs  \elin + 2 \rlin + 4 \rho \Olin) -\frac 1 2 \Omega \trch (\otxb-\otx)  \label{Db-otx}\\
\partial_v\otx&=  -(\Omega \trch) \otx + 2\om \otx +2(\Omega \trch)\olin \label{D-otx}\\
\partial_u\otxb&=  -(\Omega \trchb) \otxb + 2\omb \otxb +2(\Omega \trchb)\oblin \label{Db-otxb}\\
\ns_3(\Omega^{-1} \chiblin)&+\Omega^{-1}(\trchb)\chiblin=  -\Omega^{-1}\aalin  \label{nabb-3-chibh}\\
\ns_4(\Omega^{-1} \chilin)&+\Omega^{-1}(\trch) \chilin=  -\Omega^{-1}\alin \label{nabb-4-chih}\\
\ns_3(\Omega\chilin)&+\frac 1 2 (\Omega \trchb) \chilin +\frac 12 (\Omega \trch) \chiblin   =   -2 \Omega\slashed{\mathcal{D}}_2^\star \elin  \label{nabb-3-chih}\\
\ns_4(\Omega\chiblin)&+\frac 1 2 (\Omega \trch) \chiblin +\frac 12 (\Omega \trchb) \chilin   =   -2 \Omega\slashed{\mathcal{D}}_2^\star \eblin  \label{nabb-4-chibh}\\
 \ns_3 \eblin&= \frac 1 2 (\trchb) (\elin-\eblin) + \bblin  + \rhoF\bbFlin \label{nabb-3-eblin}\\
 \ns_4 \elin&= -\frac 1 2 (\trch) (\elin-\eblin) - \blin- \rhoF\bFlin  \label{nabb-4-elin}\\
\ns_4 \eblin&= \frac{2}{\Omega}\ns\olin- \trch\eblin + \blin  + \rhoF\bFlin \label{nabb-4-eblin}\\
\ns_3 \elin&= \frac{2}{\Omega}\ns\olinb- \trchb\elin - \bblin  - \rhoF\bbFlin \label{nabb-3-elin}\\
\slashed{\curl} \elin&=\slin, \qquad \slashed{\curl}\eblin=-\slin\label{eq:curl-elin-slin}\\
 \partial_v \oblin&= -\Omega^2\big(\rlin+2\rhoF \rhoFlin +2(\rho+\rhoF^2) \Olin\big)\label{eq:dvoblin}\\
\partial_u \olin&= -\Omega^2\big(\rlin +2\rhoF \rhoFlin+2(\rho+\rhoF^2) \Olin\big)\\
  \divs  \chiblin&= -\frac 1 2(\trchb) \elin +\bblin -\rhoF\bbFlin +\frac {1}{ 2\Omega} \nabb\otxb  \label{Codazzi-chib}\\
  \divs  \chilin&= -\frac 1 2(\trch) \eblin -\blin +\rhoF\bFlin +\frac {1}{ 2\Omega} \nabb\otx\label{Codazzi-chi}\\
\Klin&=  -\rlin+2\rhoF \rhoFlin -\frac 1 4 \frac{\trch}{\Omega} (\otxb-\otx)+\frac 1 2 \Olin (\trch \trchb),\label{eq:Gauss}
\end{align}
\paragraph{Linearised Maxwell Equations:}
\begin{align}
\ns_3 \bFlin+\Big(\frac 1 2 \trchb+\hat{\omb}\Big) \bFlin&=  -\DDs_1(\rhoFlin, \sigmaFlin) +2\rhoF\elin \label{nabb-3-bF}\\
\ns_4 \bbFlin+\Big(\frac 1 2 \trch +\hat{\om} \Big)\bbFlin&= \DDs_1(\rhoFlin, -\sigmaFlin)-2\rhoF\eblin\label{nabb-4-bbF}\\
\ns_3 \rhoFlin +\trchb \rhoFlin &=  -\divs  \bbFlin -\frac{\rhoF}{\Omega} \otxb \label{eq:nabb3-rhoFlin}\\
\ns_4 \rhoFlin +\trch \rhoFlin &=  \divs  \bFlin -\frac{\rhoF}{\Omega} \otx  \label{eq:nabb4-rhoFlin}\\
\ns_3 \sigmaFlin +\trchb \sigmaFlin &=  \slashed{\curl} \bbFlin \label{eq:nabb3-sigmaFlin} \\
\ns_4 \sigmaFlin +\trch \sigmaFlin &=  \slashed{\curl} \bFlin ,\label{eq:nabb4-sigmaFlin}
\end{align}
\paragraph{Linearised Bianchi Identities:}
\begin{align}
 \ns_3\alin+\Big(\frac 1 2 \trchb+2\hat{\omb} \Big)\alin&= -2 \DDs_2\, \blin-3\rho\chilin -2\rhoF \ \Big(\DDs_2\bFlin +\rhoF\chilin \Big) \label{nabb-3-a}\\
 \ns_4\aalin+\Big(\frac 1 2 \trch+2\hat{\om} \Big)\aalin&= 2 \DDs_2\, \bblin-3\rho\chiblin +2 \rhoF \ \Big(\DDs_2\bbFlin -\rhoF\chiblin \Big) , \label{nabb-4-aa}\\
 \ns_3 \blin+\left(\trchb +\hat{\omb}\right) \blin &= \DDs_1(-\rlin, \slin) +3 \rho \elin +\rhoF\Big(-\DDs_1(\rhoFlin, \sigmaFlin)  - \trch\bbFlin-\frac 1 2 \trchb \bFlin    \Big) , \label{nabb-3-b}\\
 \ns_4 \bblin+\Big(\trch +\hat{\om}\Big) \bblin &= \DDs_1(\rlin, \slin) -3 \rho \eblin +\rhoF\left(\DDs_1(\rhoFlin, -\sigmaFlin)  - \trchb\bFlin-\frac 1 2 \trch \bbFlin    \right),  \label{nabb-4-bb}\\
 \ns_3 \bblin+ \left( 2\trchb  -\hat{\omb}\right) \bblin &= -\divs \aalin  +\rhoF\Big(\ns_3\bbFlin-\hat{\omb} \bbFlin\Big), \label{nabb-3-bb}\\
 \ns_4 \blin+ \left(2\trch  -\hat{\om}\right)\blin &= \divs \alin  +\rhoF\left(\ns_4\bFlin-\hat{\om} \bFlin \right),\label{nabb-4-b}\\
 \ns_3 \rlin +\frac 32 \trchb \rlin &=  -\divs  \bblin-\rhoF \divs \bbFlin - \big(\frac 3 2\rho+\rhoF^2\big)\frac{1}{\Omega} \otxb-2\trchb \rhoF \rhoFlin, \label{eq:nabb-3-rlin}\\
\ns_4 \rlin +\frac 3 2 \trch \rlin &=  \divs  \blin+\rhoF \divs  \bFlin- \big(\frac 3 2\rho+\rhoF^2\big)\frac{1}{\Omega} \otx-2\trch \rhoF \rhoFlin, \label{eq:nabb-4-rlin}\\
 \ns_3 \slin+\frac 3 2 \trchb \slin&= -\slashed{\curl}\bblin - \rhoF \ \slashed{\curl}\bbFlin, \label{nabb-3-check-sigma}\\
\ns_4 \slin+\frac 3 2 \trch \slin&= -\slashed{\curl}\blin - \rhoF  \ \slashed{\curl}\bFlin. \label{nabb-4-check-sigma}
 \end{align}

 For a derivation of the system we refer to \cite{DHR, Giorgi6}.

\subsubsection{Additional Quantities}
We define the mass aspect functions as 
\begin{align}\label{eq:definition-mu}
    \mu&\doteq\divs\elin + \rlin, \qquad \underline{\mu}\doteq\divs\eblin+\rlin.
\end{align}
We also denote
\begin{align}
\lambda&\doteq \frac{r}{\Omega^2}[\otx-2\Omega\tr\chi\Olin]=\frac{r}{\Omega^2}\otx-4\Olin, \label{eq:definition-lambda} \\ \underline{\lambda}&\doteq\frac{r}{\Omega^2}[\otx-2\Omega\tr\chib\Olin]= \frac{r}{\Omega^2}\otxb+4\Olin,
\end{align}
and
\begin{align}
    X&\doteq2r^3(\rlin-\divs\eblin)+\frac{r^3\Omega\tr\chi}{2\Omega^4}[r^2\ds]\otx=2r^3(\rlin-\divs\eblin)+\frac{r^2}{\Omega^2}[r^2\ds]\otx=r^3\big(2\mu+\ds \lambda\big), \label{definition-X}\\
    \quad \quad  
    \underline{X}&\doteq2r^3(\rlin-\divs\elin)+\frac{r^3\Omega\tr\chib}{2\Omega^4}[r^2\ds]\otxb= 2r^3(\rlin-\divs\elin)-\frac{r^2}{\Omega^2}[r^2\ds]\otxb=r^3\big(2\underline{\mu}-\ds \underline{\lambda}\big).\label{definition-Xb}
\end{align}

\begin{lemma}\label{lemma:notable-transport-equations} As a consequence of the system of linearised Einstein-Maxwell equations, the above quantities satisfy
    \begin{align}
    \partial_v\lambda &= - \otx,\label{eq:partial-v-lambda} \\
      \partial_v\Big(\frac{1}{6M}X\Big)&=\Big(1 -\frac{r_c}{r}\Big)\otx-\frac{4Q}{3M}\Omega^2\rhoFlin,
    \end{align}
    where $r_c\doteq\frac{4Q^2}{3M}$ denotes the radius critical value.
\end{lemma}

\subsection{Pure Gauge Solutions}

It is well known that adopting a double null form for the metric does not fully fix the gauge freedom inherent in the Einstein equations. In particular, there exist special solutions to the linearised Einstein–Maxwell system that arise purely from infinitesimal coordinate transformations that leave the double null structure of the metric intact. These solutions, which do not correspond to physical perturbations but rather to changes in the coordinate description, are referred to as pure gauge solutions.

\begin{lemma} \label{lem:exactsol} 
For any smooth function $f=f\left(v,\theta,\phi\right)$, the following is a pure gauge solution of the system of linearised Einstein--Maxwell equations on the Reissner--Nordstr\"om background \eqref{eq:first}--\eqref{nabb-4-check-sigma}:
\begin{align}
2\Olin &= \frac{1}{\Omega^2} \partial_v \left(f \Omega^2\right)  , \qquad \frac{\glinto}{\sqrt{\slashed{g}}} = \frac{2\Omega^2 f}{r} + \frac{2}{r} r^2 \slashed{\Delta}f ,\qquad  \glinh= - \frac{4}{r} r^2 \slashed{\mathcal{D}}_2^\star \slashed{\nabla} f , \qquad \bmlin =  -2r^2 \slashed{\nabla}_A \left[ \partial_v \left(\frac{f}{r}\right)\right] , \nonumber 
\end{align}
\begin{align}
  \elin = \frac{\Omega^2}{r^2} r \slashed{\nabla} f  , \qquad \eblin = \frac{1}{\Omega^2} r \slashed{\nabla} \left[\partial_v \left(\frac{\Omega^2}{r}f\right) \right],\qquad\chiblin = -2\frac{\Omega}{r^2} r^2 \slashed{\mathcal{D}}_2^\star \slashed{\nabla} f   \nonumber  ,
\nonumber 
\end{align}
\begin{align}
\otx = 2 \partial_v \left(\frac{f \Omega^2}{r}\right)  , \qquad \otxb =  \frac{2\Omega^2}{r^2} \left[\Delta_{S^2} f + f \left(1-\frac{4M}{r}+\frac{3Q^2}{r^2} \right) \right]  , \nonumber 
 \end{align}
 \begin{align}
\rlin = \left(\frac{6M}{r^4}-\frac{8Q^2}{r^5}\right) \Omega^2 f  ,\qquad  \bblin= \left(\frac{6M\Omega}{r^4} -\frac{6Q^2\Omega}{r^5} \right)r \slashed{\nabla} f ,\qquad  \Klin =- \frac{\Omega^2}{r^3}\left(\Delta_{S^2} f + 2f\right) \nonumber 
\end{align}
\begin{align}
\rhoFlin= -\frac{2Q}{r^3}\Omega^2f, \qquad \bbFlin=-\frac{2Q\Omega}{r^3}  r \slashed{\nabla} f \nonumber
\end{align}
and
\[
\chilin = \alin = \aalin = 0 \ \ \ , \ \ \ \blin= \bFlin = 0 \ \ \ , \ \ \ \slin=\sigmaFlin  = 0 \nonumber \, .
\]
We will call $f$ a gauge function. If a quantity vanishes for such a gauge function then we say it is $f(v)$-gauge invariant. 
\end{lemma}
\begin{proof}
    The proof of Lemma 6.1.1 in \cite{DHR} can be easily adapted to the case of linearised Einstein-Maxwell equations.
\end{proof}
\begin{remark}\label{rem:fugaugeinvariant}
    One can obtain a similar lemma to the above for a function $f(u,\theta,\phi)$ and make a definition of $f(u)$-gauge invariance. 
\end{remark}

 \subsection{The Teukolsky Variables in Linear Perturbations of Reissner-Nordstr\"om}

In the vacuum case, the study of linear perturbations has traditionally relied on analyzing the decoupled Teukolsky equations \cite{Teukolsky}, which govern the evolution of certain components of the curvature, namely $\alin$ and $\ablin$. However, when considering perturbations of the Reissner-Nordström spacetime, a charged, nonvacuum solution, the situation becomes more intricate. In this context, the curvature components $\alin$ and $\ablin$ no longer evolve independently; instead, they are coupled to electromagnetic quantities, reflecting the interaction between the gravitational and electromagnetic fields.

 In \cite{Giorgi4, Giorgi5}, in addition to $\alin$, $\ablin$, the third author has identified a set of three pairs of gauge invariant quantities which satisfy a system of coupled Teukolsky-like equations. The Teukolsky variables are defined as
\begin{align}
\fflin&\doteq \DDs_2 \bFlin +\rhoF \chilin, \qquad \underline{\fflin}\doteq \DDs_2 \bbFlin -\rhoF \chiblin, \label{eq:definitions-ff}\\
\bflin&\doteq 2\rhoF \blin - 3 \rho \bFlin, \qquad \underline{\bflin}\doteq 2\rhoF \bblin -3\rho \bbFlin,\label{eq:definitions-bb}
\end{align} 
and
\begin{align}
\xflin \doteq \ns_4 \bFlin +\big( \frac 3 2 \trch -\hat{\omega}\big)\bFlin ,\qquad \xfblin \doteq \ns_3 \bbFlin +\big(\frac 3 2 \trchb-\hat{\omb}\big) \bbFlin .\label{eq:definitions-xflin}
\end{align}
The quantities $\alin, \aalin, \fflin, \underline{\fflin}, \bflin, \underline{\bflin}, \xflin, \xfblin$ satisfy the following relations \cite{Giorgi4, Giorgi5}:
\begin{align} 
\rhoF \left(\ns_3 \a+\left( \frac 1 2 \trchb - 4\omb \right)\a \right)&= -\DDs_2\bflin -\left(2\rhoF^2+3\rho\right)\fflin, \\
-\rhoF \left(\ns_4 \aa+\left( \frac 1 2 \trch - 4\om \right)\aa \right)&= -\DDs_2\underline{\bflin} -\left(2\rhoF^2+3\rho\right)\underline{\fflin}, \label{eq:Bianchi-aa-ff-bb}
\end{align}
\begin{align*}
\ns_4\fflin + 2 \trch \fflin&= -\rhoF \alin + \DDs_2 \xflin,\\
\ns_3\underline{\fflin} + 2 \trchb \underline{\fflin}&= \rhoF \aalin + \DDs_2 \xfblin, \\
\ns_4 \bflin+ 3 \trch \bflin&=  \rhoF \div \alin - \big( 3\rho -2\rhoF^2\big) \xflin, \\
\ns_3 \underline{\bflin}+ 3 \trchb \underline{\bflin}&=  -\rhoF \div \aalin - \big( 3\rho -2\rhoF^2\big) \xfblin\\ 
\ns_3 \xflin+\frac 1 2 \trchb \xflin &=  - \div \fflin - 2\bflin, \\
\ns_4 \xfblin+\frac 1 2 \trch \xfblin &=  - \div \underline{\fflin} - 2\underline{\bflin}.
\end{align*} 
As a consequence of the above, the quantities $\alin, \aalin, \fflin, \underline{\fflin}, \bflin, \underline{\bflin}$ can be shown to satisfy a system of coupled Teukolsky-like wave equations,
\begin{align} 
\ns_4 \ns_3(r\Omega^2 \alin)-\Big(\frac{2\omega}{\Omega}-\frac{4\Omega}{r}\Big)\ns_3(r\Omega^2 \alin) +\Ds _2\divs[ r\Omega^2\alin]-\left(3\rho-2\rhoF^2\right)r \Omega^2\alin&=-\frac{4Q\Omega^2}{r}\ns_4\fflin,\label{eq:teuk-eq-a}\\
    \ns_3\ns_4(r^4\Omega\fflin)-\frac{\omega}{\Omega}\ns_3[r^4\Omega\fflin]-\Big(\frac{\omega}{\Omega}+ \frac{2\Omega}{r}\Big)\ns_4[\Omega r^4\fflin]+\Ds _2\divs[r^4\Omega\fflin]+\left(5\rhoF^2+6\rho\right)r^4\Omega\fflin&=\frac{r^2\rhoF}{\Omega}\ns_3(r^2\Omega^2\alin),\\
    \ns_3\ns_4 (r^6\Omega\bflin)-\frac{\omega}{\Omega}\ns_3 [r^6\Omega\bflin]-\Big(\frac{\omega}{\Omega}+ \frac{2\Omega}{r}\Big)\ns_4 (r^6\Omega\bflin)+\divs\Ds _2[r^6\Omega\bflin]-3\big(\rho+\rhoF^2\big) r^6\Omega\bflin&=\frac{4r^4\rhoF^2}{\Omega}\ns_3(r ^2\Omega^2\xflin),
\end{align}
and 
\begin{align}
    &\ns_4\ns_3(r \Omega^2\xflin)- \frac{4\Omega}{r}\ns_3(r\Omega^2\xflin)+\divs\Ds _2[r\Omega^2\xflin]-r\big( 3\rho -2\rhoF^2\big) \Omega^2\xflin= - \frac{\Omega^2}{r^5}\ns_4(r^8\bflin),\label{eq:teuk-eq-xx}
\end{align}
as originally shown in \cite{Giorgi4, Giorgi5}. Decay estimates for these Teukolsky variables have been obtained in \cite{Giorgi7a, Apetroaie24} for $|Q|<M$ and $|Q|=M$ respectively.

\subsection{The Class of Solutions}

We consider smooth solutions $\mathcal{S}$ to the system of linearised Einstein--Maxwell equations on a fixed Reissner--Nordstr\"om background. These solutions arise from prescribing characteristic initial data on a pair of intersecting null cones:
\begin{align*} 
C_{u_0}=\{ u_0\} \times \{v \geq v_0\} \times S^2 \qquad \qquad \underline{C}_{v_0}= \{u \geq u_0\}  \times \{ v_0\}  \times S^2,
\end{align*} 
for some fixed values $u_0, v_0$. These null cones, which represent outgoing and incoming null hypersurfaces respectively, are illustrated in Figure~\ref{figure1}.

Even after imposing the double null gauge, one retains a residual gauge freedom in the Einstein equations. In particular, one can add pure gauge solutions—solutions generated by infinitesimal coordinate transformations that preserve the double null form of the metric—to modify the initial data. This freedom can be leveraged to impose certain desirable conditions on the initial data along the hypersurfaces $C_{u_0} \cup \underline{C}_{v_0}$.

To formalize this, we adopt the notion of a partially initial data normalised solution, introduced in \cite{HolCL16} and further developed in \cite{ColHol24}. In contrast to the vacuum case, where one often restricts attention to perturbations supported on angular frequencies $\ell\geq 2$, in the Einstein–Maxwell setting this is not sufficient: electromagnetic radiation contributes already at the level of $\ell=1$ (see \cite{Giorgi6}). Thus, our analysis must include modes with $\ell\geq 1$.

\begin{definition}\label{def:initial-data} We say that a solution $\mathcal{S}$ is a partially initial data normalised solution supported on $\ell \geq 1$ if it satisfies the following: \begin{itemize} \item $\mathcal{S}$ is supported on angular modes with $\ell \geq 1$, 
\item The initial data prescribed on the union of null cones $C_{u_0} \cup \underline{C}_{v_0}$
 satisfies \begin{align*} &\otx(\infty, v_0, \theta, \phi) = 0, \quad (\divs \elin + \rlin)(\infty, v_0, \theta, \phi) = 0 \quad \text{on } S^2_{\infty, v_0}, \ &r^2 \Klin(u_0, \infty, \theta, \phi) = 0. \end{align*} \end{itemize} \end{definition}

Note that partial initial data normalisation as above can be achieved for any Reissner–Nordström spacetime in the subextremal range $|Q|< M$.

It is worth noting that the first two conditions in the definition above are preserved under evolution along the event horizon, as can be seen directly from the evolution equations~\eqref{D-otx},~\eqref{nabb-4-elin} and~\eqref{eq:nabb-4-rlin}. Similarly, the third condition involving $\Klin$ is propagated along future null infinity.

To control the asymptotic behaviour of solutions, we impose mild decay assumptions toward null infinity. Specifically, we assume that the following weighted components of $\mathcal{S}$ admit well-defined, finite limits as $r\to \infty$, for some $s \in (0,1)$:
\begin{equation}\label{decreten} 
\begin{split}
r^{3+s} \alin  , \quad r^{3+s}\blin  , \quad   r^3\rlin , \quad r^3 \slin  , \quad  r^2 \bblin , \quad r\ablin , \quad r^{2+s}\bFlin,\quad r^2\rhoFlin,  \quad r^2\sigmaFlin ,\quad r\bbFlin, \\
 r^2 \xlin , \quad r\xblin , \quad r \elin , \quad  r^2 \eblin , \quad  r^2 \divs  \elin , \quad r^3 \divs  \eblin  , \quad  r^{2+s} \olin , \quad  \olinb , \quad   r^2 \otx  , \quad  r \otxb  , \quad \Olin , \quad  r^3 \Klin.   
\end{split}
\end{equation}

Furthermore, letting $\mathcal{Q}$ denote any of the above quantities, we require that for all multi-indices $i+j\leq 2$, the derivatives $[r\slashed\nabla]^i[\slashed\nabla_T]^i\mathcal{Q}$ also admit well-defined limits at null infinity, and that for each fixed $u_1>u_0$, the supremum estimate
\begin{align} \label{decreten2}
\sum_{i+j\leq 2} \sup_{\left[u_0,u_1\right] \times \{ v\geq v_0\} \times S^2} |[r\slashed{\nabla}]^i [\slashed{\nabla}_T]^j \mathcal{Q}| \leq C\left[u_1\right]
\end{align}
holds, where $C\left[u_1\right]$ depends only on $u_1$ and the initial data but not on $v$.

These assumptions are compatible with the well-posedness theory developed for gravitational perturbations in \cite{DHR}, and they extend naturally to the Einstein–Maxwell setting considered here.

Finally, for regularity at the event horizon $\mathcal{H}^+ \cap \{ v \geq v_0\}$, we recall that the following set of renormalised quantities extends smoothly:
\begin{align*} 
 \Olin, \quad \otx, \quad \Omega^{-2} \otxb, \quad \olin, \quad \Omega^{-2} \oblin, \quad \elin, \quad \eblin, \quad \Omega \chilin, \quad \Omega^{-1} \chiblin, \quad \Klin,\\
 \Omega \blin, \quad \Omega^{-1} \bblin, \quad \Omega^2 \alin, \quad \Omega^{-2} \aalin, \quad \rlin, \quad \slin, \quad \Omega \bFlin, \quad \Omega^{-1} \bbFlin, \quad \rhoFlin, \quad \sigmaFlin.  
\end{align*} 
\begin{lemma}\label{lemma:lambda-X}
    Suppose $\mathcal{S}$ is partially initial data normalised and supported on $\ell\geq 1$. Then, as $u\to \infty$ along the initial cone $\underline{C}_{v_0}$, for $|Q|<M$, we have \begin{align*}
        \lambda(u,v_0,\theta,\phi) \ =\  \mathcal{O}(\Omega^2) \qquad \text{and} \qquad X(u,v_0,\theta,\phi) \ =\  \mathcal{O}(\Omega^2)
    \end{align*} 
    In particular, $\Omega^{-2}\lambda$ and $\Omega^{-2}X$ extend smoothly to the event horizon $\mathcal{H}^+ \cap \{ v \geq v_0\}$. 
\end{lemma}
\begin{proof}
Recall the definition of $\mu$ and $\lambda$ in \eqref{eq:definition-mu}-\eqref{eq:definition-lambda}.  Using equations (\ref{eq:pr-Olin}), (\ref{Db-otx}) we obtain \begin{align*}
      \partial_u \big( r\otx - 4\Omega^2 \Olin)\big) &= \left(4\rho \Omega^2 + 8\partial_u(\Omega)\right)\Olin  + 2\Omega^2(r\divs\elin + r\rlin) - \Omega^4 \frac{\otxb}{\Omega^2} + 4\Omega^4 \frac{\oblin}{\Omega^2} \\
      &=2\Omega^2r\mu - \Omega^4 \frac{\otxb}{\Omega^2} + 4\Omega^4 \frac{\oblin}{\Omega^2}.
  \end{align*}
Using equations \eqref{nabb-3-elin} and \eqref{eq:nabb-3-rlin} we obtain
\begin{align*}
      \partial_u(r\mu)\  = \ & 2\Omega^2 \mu - 2\Omega^2 r\divs \Omega^{-1}\bblin - 2\Omega^2 \rhoF r\divs \Omega^{-1} \bbFlin + 4\Omega^2 \rhoF \rhoFlin \\ &- (\frac{3}{2}\rho + \rhoF^2) r\Omega^2 \cdot\frac{\otxb}{\Omega^2} + \frac{2}{r}\Omega^2 [r^2\ds] \frac{\oblin}{\Omega^2} \ \doteq\  \Omega^2 \cdot R  \end{align*}
      where $R$ is regular at the event horizon. 
     Thus, in view of $\mu|_{\mathcal{H}^+} =0 $ in the partially initial data normalised gauge, and $d\Omega = (-\frac{M}{r^2}+\frac{Q^2}{r^3}) \Omega du$, after we integrate along the initial cone $\underline{C}_{v_0}$ we obtain \begin{align*}
         \int_u^\infty \partial_u(r\mu) (\bar{u},v_0,\theta,\phi) d\bar{u} =    \int_u^\infty \Omega^2 R(\bar{u},v_0,\theta,\phi) d\bar{u} = \int_{\Omega}^0 \Omega \cdot \frac{R}{\big(-\frac{M}{r^2}+\frac{Q^2}{r^3}\big)}d\Omega 
      \end{align*}
      or, \begin{align*}
          |r\mu|(u,v_0,\theta,\phi) \lesssim_Q  \sup_{u\leq\bar{u}\leq \infty}| R(\bar{u}) | \int_0^{\Omega}\Omega d\Omega
      \end{align*}
for $Q<M$,    which shows that $r\mu(u,v) = \mathcal{O}(\Omega^2)$, along $\underline{C}_{v_0}$, as $u\to \infty$. 
Therefore, going back to the expression of $\Omega^2\lambda$ we have \begin{align*}
           \partial_u \big( r\otx - 4\Omega^2 \Olin)\big) = \mathcal{O}(\Omega^4), \qquad \text{along} \quad \underline{C}_{v_0}, \quad \text{as} \quad u \to \infty,
      \end{align*}
      from which we deduce that $\lambda(u,v) = \mathcal{O}(\Omega^2)$, along $\underline{C}_{v_0}$, as $u\to \infty$. 
      Lastly, as in (\ref{definition-X}) we write  \begin{align*}
          X = r^3 (2\mu + \ds\lambda) 
      \end{align*}
     from which we readily obtain using previous results that $X(u,v_0,\theta,\phi) =  \mathcal{O}(\Omega^2)$, along $\underline{C}_{v_0}$, as $u\to \infty$. 
\end{proof}

\section{The Conservation Law and the Basic Coercive Estimate}
\subsection{The Conservation Law}

In this section, we state the novel conservation law in the case of Einstein-Maxwell equations, which extends the one known in vacuum \cite{HolCL16,Col23}.

Let $\mathcal{S}$ be a smooth solution of the system of linearised Einstein-Maxwell equations.
For any $-\infty\leq u_0\leq u_1 \leq \infty$ and $-\infty\leq v_0\leq v_1 \leq \infty$ we define the following fluxes:
\begin{align}
E_{u}[\mathcal{S}](v_0,v_1)&\doteq \int_{v_0}^{v_1}\int_{S^2}\Big[|\Omega\xlin|^2+2|\Omega\eblin|^2+2|\Omega\bFlin|^2+2|\Omega\sigmaFlin|^2+2|\Omega\rhoFlin|^2+r\otx(\rlin+\divs\eblin)\Big]r^2 dv \varepsilon_{S^2},\label{eq:def-Eu}\\
E_{v}[\mathcal{S}](u_0,u_1)&\doteq \int_{u_0}^{u_1}\int_{S^2}\Big[|\Omega\xblin|^2+2|\Omega\elin|^2+2|\Omega\bbFlin|^2+2|\Omega\sigmaFlin|^2+2|\Omega\rhoFlin|^2-r\otxb(\rlin+\divs\elin)\Big]r^2 du \varepsilon_{S^2}.\label{eq:def-Ev}
\end{align}

For these fluxes we have the following conservation law:
\begin{proposition}\label{prop:conservationlaws}
For any $-\infty\leq u_0\leq u_1 \leq \infty$ and $-\infty\leq v_0\leq v_1 \leq \infty$ we have the conservation law
    \begin{align}\label{eq:CL}
E_{u_1}[\mathcal{S}](v_0,v_1)+E_{v_1}[\mathcal{S}](u_0,u_1)&=E_{u_0}[\mathcal{S}](v_0,v_1)+E_{v_0}[\mathcal{S}](u_0,u_1).
\end{align}
\end{proposition}
\begin{proof}
    Direct computation as in Proposition 4.2 in \cite{HolCL16}.
\end{proof}

\begin{remark}
We note that the fluxes
\begin{align*}
E^{\mathrm{aux}}_u[\mathcal{S}](v_0,v_1)&\doteq\int_{v_0}^{v_1}\int_{S^2}\Big[\frac{1}{2}\otx^2+r\otx(\rlin+\divs\eblin)+2\olin\otxb-4\omega\Olin\otx\Big]r^2dv\varepsilon_{S^2},\\
E^{\mathrm{aux}}_v[\mathcal{S}](u_0,u_1)&\doteq\int_{u_0}^{u_1}\int_{S^2}\Big[\frac{1}{2}\otxb^2-r\otxb(\rlin+\divs\elin)+2\olinb\otx+4\omega\Olin\otxb\Big]r^2du\varepsilon_{S^2},
\end{align*}
also satisfy the conservation law
\begin{align*}
    E^{\mathrm{aux}}_{u_1}[\mathcal{S}](v_0,v_1)+E^{\mathrm{aux}}_{v_1}[\mathcal{S}](u_0,u_1)&=E^{\mathrm{aux}}_{u_0}[\mathcal{S}](v_0,v_1)+E^{\mathrm{aux}}_{v_0}[\mathcal{S}](u_0,u_1).
\end{align*}
    This conservation law can be combined with~\eqref{eq:CL} to recover the form from Holzegel's original work~\cite{HolCL16}, which excludes curvature. 
    \end{remark}
\begin{remark}
There are further conservation laws encoded in the linearised Einstein--Maxwell equations. In particular, as in Schwarzschild there is a conservation law at the level of curvature involving $\blin$ and $\bblin$. As was shown in~\cite{Col23}, this arises from manipulating the conservation law resulting from commutating~\eqref{eq:CL} with the angular Killing symmetries. Here, we avoid deriving the exact conservation law but instead produce control by using the commutation of~\eqref{eq:CL}.
\end{remark}

\subsection{The Coercive Estimate}

The fluxes $E_{u}[\mathcal{S}](v_0,v_1)$ and $E_{v}[\mathcal{S}](u_0,u_1)$ are not obviously coercive, but can be manipulated in a form where coercivity is manifest. 

Define the following (coercive) master energies for $\mathcal{S}$ along the final outgoing cone $C_{u_1}$: 
\begin{align}\label{eq:definition-master-energy}
\begin{split}
    \mathbb{E}^{i,j}[\mathcal{S}](u_1)&\doteq \sum_{k\leq i,m\leq j}\int_{v_0}^{\infty}\int_{S^2}\Big[|[\nabla_T]^k[r\ns]^m\xlin|^2+\Big|[\nabla_T]^k[r\ns]^m\Big[\eblin-[r\ns]\frac{\otx}{\Omega^2}\Big]\Big|^2\Big]\Omega^2r^2dv\varepsilon_{S^2}\Big|_{u=u_1}\\
        & +\sum_{k\leq i,m\leq j}\int_{v_0}^{\infty}\int_{S^2}\Big[|[\nabla_T]^k[r\ns]^m\bFlin|^2+|[\nabla_T]^k[r\ns]^m\sigmaFlin|^2\Big|_{u=u_1} \\
 & +\sum_{k\leq i,m\leq j}\int_{v_0}^{\infty}\int_{S^2}\Big[\Big|[\nabla_T]^k[r\ns]^m\Big[\rhoFlin-\frac{Q\lambda}{r^2}\Big]\Big|^2\Big]\Omega^2r^2dv\varepsilon_{S^2}\Big|_{u=u_1} +\sum_{k\leq i,m\leq j}\|[\nabla_T]^k[r\ns]^mX\|^2_{u_1,\infty}.
    \end{split}
\end{align}

\begin{theorem}\label{thm:mastercontrol}
Let $\mathcal{S}$ a partially initial data normalized solution supported on $\ell \geq 1$ of the system of linearised Einstein--Maxwell equations that satisfies the conditions~\eqref{decreten}-\eqref{decreten2} and $|Q|<\frac{\sqrt{15}}{4}M$. Then 
\begin{align*}
    \mathbb{E}^{0,0}[\mathcal{S}](u_1)+\lim_{v\to \infty }\int_{u_0}^{u_1}\int_{S^2}\Big[|r\Omega\xblin|^2+|r\Omega\bbFlin|^2\Big] du \varepsilon_{S^2}\lesssim \mathbb{E}_{data}[\mathcal{S}](u_1)
\end{align*}
where, 
\begin{align*}
    \mathbb{E}_{data}[\mathcal{S}](u_1)\doteq E_{u_0}[\mathcal{S}](v_0,\infty)+E_{v_0}[\mathcal{S}](u_0,u_1)+\|\Omega^{-1}\lambda\|_{u_1,v_0}^2+\|\Omega^{-1}X\|_{u_1,v_0}^2,
\end{align*}
with $\sup_{u\in[u_0,\infty)}\mathbb{E}_{data}[\mathcal{S}](u)<\infty$ for regular initial data.  
\end{theorem}
\begin{remark}
    Theorem~\ref{thm:mastercontrol} can be proved using the method of~\cite{HolCL16}. In that work, Holzegel uses the residual gauge freedom of Lemma~\ref{lem:exactsol} to change gauge only on the final outgoing cone to achieve $\otx[\mathcal{S}-\mathcal{S}_{pg}]=0$, where $\mathcal{S}_{pg}$ denotes a pure gauge solution. Observe that $E_{u}[\mathcal{S}]$ is then manifestly coercive and $E_{u}^{\mathrm{aux}}[\mathcal{S}]$ vanishes. The price to pay is an estimate of the difference $E_{u}[\mathcal{S}]-E_{u}[\mathcal{S}-\mathcal{S}_{pg}]$, which is a boundary term involving the pure gauge solution. This boundary term can be controlled by the flux. Since $\xlin$ is $f(v)$-invariant, one obtains control of it in the initial data gauge. The argument in this work refines the procedure for obtaining a coercive estimate by avoiding the use of gauge manipulations on the final outgoing cone.  
\end{remark}
\begin{proof}[Proof of Theorem~\ref{thm:mastercontrol}]
The starting point for the coercive estimate is the conservation law
\begin{align*}
    E_{u_1}[\mathcal{S}](v_0,v_1)+E_{v_1}[\mathcal{S}](u_0,u_1)=E_{u_0}[\mathcal{S}](v_0,v_1)+E_{v_0}[\mathcal{S}](u_0,u_1).
\end{align*}
We now manipulate the fluxes $E_{u_1}[\mathcal{S}](v_0,v_1)$ and $E_{v_1}[\mathcal{S}](u_0,u_1)$ into manifestly $f(v)$-gauge invariant and $f(u)$-gauge invariant forms respectively (see Lemma~\ref{lem:exactsol} and Remark~\ref{rem:fugaugeinvariant}). To do this we write
    \begin{align*}
    r^3(\rho+\divs\eblin)&=\frac{1}{2}X-\frac{r^2}{2\Omega^2}[r^2\ds]\otx+2r^3\divs\eblin,\\
    r^3(\rho+\divs\elin)&=\frac{1}{2}\Xb+\frac{r^2}{2\Omega^2}[r^2\ds]\otxb+2r^3\divs\elin.
\end{align*}
This yields,
\begin{align*}
E_{u_1}[\mathcal{S}](v_0,v_1)&=\int_{v_0}^{v_1}\int_{S^2}\Big[|r\Omega\xlin|^2+2|r\Omega\bFlin|^2+2|r\Omega\sigmaFlin|^2+2|r\Omega\rhoFlin|^2+2\Big|r\Omega\eblin-\frac{1}{2}[r\ns]\frac{r\otx}{\Omega}\Big|^2\Big]dv\varepsilon_{S^2}\\
&\qquad +\int_{v_0}^{v_1}\int_{S^2}\frac{1}{2}X\otx dv\varepsilon_{S^2},\\
    E_{v_1}[\mathcal{S}](u_0,u_1)&\doteq \int_{u_0}^{u_1}\int_{S^2}\Big[|r\Omega\xblin|^2+2|r\Omega\bbFlin|^2+2|r\Omega\sigmaFlin|^2+2|r\Omega\rhoFlin|^2+2\Big|r\Omega\elin+\frac{1}{2}[r\ns]\frac{r\otxb}{\Omega}\Big|^2\Big] du \varepsilon_{S^2}\\
    &\qquad+\int_{u_0}^{u_1}\int_{S^2}\Big[-\frac{1}{2}\Xb \otxb\Big] du \varepsilon_{S^2}.
\end{align*}
Note that all terms, except $\rhoFlin$-terms, in the first line of each flux are $f(v)$- or $f(u)$-gauge invariant respectively. 

In fact, one can write the last terms as pure boundary terms as follows. From Lemma \ref{lemma:notable-transport-equations}, we know that the quantities $X$ and $\Xb$ obey nice transport equations,
\begin{align*}
       \partial_vX&=6M\Big(1 -\frac{r_c}{r}\Big)\otx-8Q\Omega^2\rhoFlin\\
       \partial_u\Xb&=6M\Big(1 -\frac{r_c}{r}\Big)\otxb+8Q\Omega^2\rhoFlin
    \end{align*}
    in any gauge\footnote{Note that in Schwarzschild ($Q=0$) this transport equation allows one can simply write $\partial_v\frac{X^2}{6M}=X\otx$ to obtain a coercive energy on the outgoing cone.}. So, one can compute 
    \begin{align*}
        \frac{1}{2}X\otx&=\partial_v\Big[-\Big(1-\frac{r_c}{r}\Big)\frac{3M}{2}\lambda^2-\frac{1}{2}X\lambda\Big]+2\Big|r\Omega\rhoFlin-\frac{Q\Omega\lambda}{r}\Big|^2-2|r\Omega\rhoFlin|^2,\\
        \frac{1}{2}\Xb\otxb&=\partial_u\Big[-\Big(1-\frac{r_c}{r}\Big)\frac{3M}{2}\lb^2+\frac{1}{2}\Xb\lb\Big]-2\Big|r\Omega\rhoFlin+\frac{Q\Omega\lb}{r}\Big|^2+2|r\Omega\rhoFlin|^2.
    \end{align*}
    Substituting in the above, we obtain
    \begin{align}
E_{u_1}[\mathcal{S}](v_0,v_1)&=\int_{v_0}^{v_1}\int_{S^2}\Big[|r\Omega\xlin|^2+2|r\Omega\bFlin|^2+2|r\Omega\sigmaFlin|^2+2\Big|r\Omega\rhoFlin-\frac{Q\Omega\lambda}{r}\Big|^2+2\Big|r\Omega\eblin-\frac{1}{2}[r\ns]\frac{r\otx}{\Omega}\Big|^2\Big]dv\varepsilon_{S^2} \nonumber\\
&\qquad +\int_{S^2}\Big[-\Big(1-\frac{r_c}{r}\Big)\frac{3M}{2}\lambda^2-\frac{1}{2}X\lambda\Big]\varepsilon_{S^2}\Big|_{(u_1,v_0)}^{(u_1,v_1)},\label{eq:E_u1-manipulated}\\
    E_{v_1}[\mathcal{S}](u_0,u_1)&= \int_{u_0}^{u_1}\int_{S^2}\Big[|r\Omega\xblin|^2+2|r\Omega\bbFlin|^2+2|r\Omega\sigmaFlin|^2+2\Big|r\Omega\rhoFlin+\frac{Q\Omega\lambda}{r}\Big|^2+2\Big|r\Omega\elin+\frac{1}{2}[r\ns]\frac{r\otxb}{\Omega}\Big|^2\Big] du \varepsilon_{S^2}\nonumber\\
    &\qquad+\int_{S^2}\Big[\Big(1-\frac{r_c}{r}\Big)\frac{3M}{2}\lb^2-\frac{1}{2}\Xb\lb\Big]\varepsilon_{S^2}\Big|_{(u_0,v_1)}^{(u_1,v_1)}.
\end{align}
In this form, energies $E_{v_1}[\mathcal{S}](u_0,u_1)$ and $E_{u_1}[\mathcal{S}](v_0,v_1)$ are manifestly coercive except boundary terms on the final sphere at $(u_1,v_1)$ and on the initial spheres at $(u_1,v_0)$ and at $(u_0,v_1)$. The latter are data terms and so are unconcerning. The argument to establish a coercive estimate from this conservation law proceeds in two steps to deal with the two boundary terms:
\begin{enumerate}
\item If one pushes $E_{v_1}[\mathcal{S}](u_0,u_1)$ to null infinity the boundary term vanishes in the limit as $v_1\rightarrow\infty$ for $u_1<\infty$. Therefore, the limit of $E_{v_1}[\mathcal{S}](u_0,u_1)$ is manifestly coercive. To establish that the boundary term vanishes, we write the Gauss equation as
     \begin{align*}
        2 r^{3}\Klin&=  2r^{3}[-\rlin+2\rhoF \rhoFlin] -\Omega^2r\lb+r^{2}\otx.
     \end{align*}
     From~\eqref{decreten} we note that the $r^{3}\Klin$, $r^3\rlin$, $r^2\rhoFlin$ and $r^2\otx$ have well-defined limits for $v_1\rightarrow \infty$ for $u_1<\infty$, we infer $r\lb$ has a well-defined limit and $\lb\rightarrow 0$ for $v_1\rightarrow \infty$ for $u_1<\infty$. One can write $r^{-1}\Xb$ as
    \begin{align*}
        r^{-1}\Xb=2r^2[\rho-\divs\elin]-\frac{r}{\Omega^2}[r^2\ds]\otxb=2r^2[\rlin+\divs\eblin]-[r^2\ds]\lb.
    \end{align*}
    From the decay of $\lb$ and the conditions~\eqref{decreten} we infer that $r^{-1}\Xb$ vanishes in the limit $v\rightarrow \infty$. This means $\Xb\lb=r^{-1}\Xb r\lb$ vanishes as $v\rightarrow \infty$. Therefore, in the limit as $v\rightarrow \infty$, 
    \begin{align*}
         E_{v_1}[\mathcal{S}](u_0,u_1)&= \int_{u_0}^{u_1}\int_{S^2}\Big[|r\Omega\xblin|^2+2|r\Omega\bbFlin|^2\Big] du \varepsilon_{S^2},
    \end{align*}
    where we use that
    \begin{align*}
        r\Omega\elin+\frac{1}{2}[r\ns]\frac{r\otxb}{\Omega}=\frac{\Omega}{2}[r\ns]\lb-\Omega r\eblin\rightarrow 0,
    \end{align*}
    as $v_1\rightarrow \infty$ for $u_1<\infty$.
    \item To deal with the boundary term on the final sphere at $(u_1,v_1)$ from the outgoing flux $E_{u_1}[\mathcal{S}](v_0,v_1)$ we manipulated it into the form:
\begin{align*}
    -\Big(1-\frac{r_c}{r}\Big)\frac{3M}{2}\lambda^2-\frac{1}{2}X\lambda=\frac{1}{24M}X^2+\frac{2Q^2}{r}\lambda^2-\frac{|X+6M\lambda|^2}{24M}.
\end{align*}
The bad boundary term $|X+6M\lambda|^2$ can be estimated as follows. From Lemma~\ref{lemma:notable-transport-equations}, we have\footnote{In Schwarzschild ($Q=0$), the equation~\eqref{eq:criticaltransporteq} implies that the boundary term can be exchanged for a manifestly positive one with no estimate required. Thus this proof simplifes the argument in~\cite{HolCL16} and~\cite{ColHol24}.}
        \begin{align}
            \partial_v\Big[\frac{1}{\sqrt{24M}}\Big(X+\Big(1-\frac{r_c}{r}\Big)6M\lambda\Big)\Big]=\Omega^2\frac{4\sqrt{2}Q}{\sqrt{24M}}\sqrt{2}\Big[\frac{Q\lambda}{r^2}- \rhoFlin\Big].\label{eq:criticaltransporteq}
        \end{align}
       Lemma~\ref{lem:tplemma} implies
        \begin{align*}
          \frac{\|X+6M\lambda\|^2_{u_1,\infty}}{24M}&\leq \frac{\big\|X+\big(1-\frac{r_c}{r(u,v_0)}\big)6M\lambda \big\|^2_{u_1,v_0}}{24M(1-\frac{r_+}{ r(u,v_0)})}+\frac{32Q^2 }{24Mr_+}\int_{v_0}^{\infty}2r^2\Omega^2\Big|\Big[\frac{Q\lambda}{r^2}-\rhoFlin\Big]\Big|^2dv\eps_{S^2}.
        \end{align*}
        where we used that $\lambda$ is regular for $u_1<\infty$ and $v_1\rightarrow \infty$. Using Lemma \ref{lemma:lambda-X} to infer that $X+\big(1-\frac{r_c}{r}\big)6M\lambda$ is smooth on the exterior and behaves as $\mathcal{O}(\Omega^2)$ asymptotically towards the horizon for $Q<M$, we deduce that the initial data boundary term on the right hand side has a limit towards the event horizon. To absorb the second term on the right hand side into the energy flux in \eqref{eq:E_u1-manipulated} requires 
        \begin{align*}
            1-\frac{1}{24M}\frac{32Q^2}{r_+}\geq 0\Longleftrightarrow  r_+\geq  r_c\Longleftrightarrow |Q|\leq \frac{\sqrt{15}}{4}M.
        \end{align*}
        This is the only step in which we see the restriction on charge.\footnote{Note if $r(u,v_0)$ is sufficiently far from the horizon then we can revisit the proof of Lemma~\ref{lem:tplemma} and take $\eps<1$. This then gives coercive fluxes. For example in the ERN case $Q=M$ we picking $\eps=\frac{3}{4}$ to close the estimate gives a restriction that $r(u,v_0)\geq \frac{4M}{3}$.}\footnote{See appendix~\ref{app:optimality} for a argument about the optimality of this restriction on charge.} Note that the strict inequality in statement of the theorem arises from requiring that one continues to control $\rhoFlin-\rhoF \lambda$.
\end{enumerate}

Note that at data we collect the contributions on $S^{2}_{u_1,v_0}$ of
\begin{align*}
    \int_{S^2}\Big[\Big(1-\frac{4Q^2}{3Mr}\Big)\frac{3M}{2}\lambda^2+\frac{1}{2}X\lambda\Big]\eps_{S^2}\Big|_{(u_1,v_0)}\qquad \text{and}\qquad \Big\|\Omega^{-1}\Big(X+\Big(1-\frac{r_c}{r}\Big)6M\lambda \Big)\Big\|^2_{u_1,v_0},
\end{align*}
which is controlled by $\mathbb{E}_{data}[\mathcal{S}](u_1)$.
\end{proof}

\begin{corollary}\label{corollary:higher-derivatives}
    Let $\mathcal{S}$ a partially initial data normalized solution supported on $\ell \geq 1$ of the system of linearised Einstein--Maxwell equations that satisfies the conditions~\eqref{decreten} and $|Q|<\frac{\sqrt{15}}{4}M$. Then 
\begin{align*}
    \mathbb{E}^{i,j}[\mathcal{S}](u_1)+\int_{u_0}^{u_1}\int_{S^2}\Big[[\ns_T]^i[r\ns]^jr\Omega\xblin|^2+\|[\ns_T]^i[r\ns]^jr\Omega\bbFlin|^2\Big] du \varepsilon_{S^2}\lesssim \mathbb{E}^{i,j}_{data}[\mathcal{S}]
\end{align*}
where, 
\begin{align*}
    \mathbb{E}^{i,j}_{data}[\mathcal{S}]\doteq E^{i,j}_{u_0}[\mathcal{S}](v_0,\infty)+E^{i,j}_{v_0}[\mathcal{S}](u_0,u_1)+\|[\ns_T]^i[r\ns]^j\Omega^{-1}\lambda\|_{u_1,v_0}^2+\|[\ns_T]^i[r\ns]^j\Omega^{-1}X\|_{u_1,v_0}^2.
\end{align*}
and
\begin{align*}
E_{u_0}^{i,j}[\mathcal{S}](v_0,v_1)&\doteq\sum_{k\leq i,m\leq j}\int_{v_0}^{v_1}\int_{S^2}\Big[|[\ns_T]^k[r\ns]^mr\Omega\xlin|^2+2|[\ns_T]^k[r\ns]^mr\Omega\bFlin|^2\Big]dv\varepsilon_{S^2}\\
&\qquad+\int_{v_0}^{v_1}\int_{S^2}\Big[\frac{1}{2}\langle[\ns_T]^i[r\ns]^jX,[\ns_T]^i[r\ns]^j\otx\rangle+2|[\ns_T]^i[r\ns]^jr\Omega\sigmaFlin|^2\\
&\qquad\qquad\qquad+2|[\ns_T]^i[r\ns]^jr\Omega\rhoFlin|^2+2\Big|[\ns_T]^i[r\ns]^j\Big[r\Omega\eblin-\frac{1}{2}[r\ns]\frac{r\otx}{\Omega}\Big]\Big|^2\Big] dv\varepsilon_{S^2},\\
    E^{i,j}_{v_0}[\mathcal{S}](u_0,u_1)&\doteq \sum_{k\leq i,m\leq j}\int_{v_0}^{v_1}\int_{S^2}\Big[|[\ns_T]^k[r\ns]^mr\Omega\xblin|^2+2|[\ns_T]^k[r\ns]^mr\Omega\bbFlin|^2\Big]dv\varepsilon_{S^2}\\
&\qquad+\int_{v_0}^{v_1}\int_{S^2}\Big[-\frac{1}{2}\langle[\ns_T]^i[r\ns]^j\Xb,[\ns_T]^i[r\ns]^j\otxb\rangle+2|[\ns_T]^i[r\ns]^jr\Omega\sigmaFlin|^2\\
&\qquad\qquad\qquad+2|[\ns_T]^i[r\ns]^jr\Omega\rhoFlin|^2+2\Big|[\ns_T]^i[r\ns]^j\Big[r\Omega\elin+\frac{1}{2}[r\ns]\frac{r\otxb}{\Omega}\Big]\Big|^2\Big] dv\varepsilon_{S^2}
\end{align*}
\end{corollary}
\begin{proof}
    This follows from Theorem~\ref{thm:mastercontrol} and the identities~\eqref{eq:LieComId}.
\end{proof}

\section{A Hierarchy of Estimates on Outgoing Cones}\label{section-hierarchy}

To prove Theorem \ref{main-theorem}, we want to establish boundedness of the energy flux associated with the Teukolsky variables $\bflin$, $\fflin$, $\alin$, $\underline{\bflin}$, $\underline{\fflin}$ and $\ablin$. Our strategy is to uncover a hierarchy within the system of gravitational and electromagnetic perturbations that enables us to leverage the control provided by Theorem~\ref{thm:mastercontrol} over the coercive master energy $\mathbb{E}^{i,j}[\mathcal{S}]$, as defined in~\eqref{eq:definition-master-energy}. In the following subsection we describe the hierarchy to prove boundedness of the energy flux for the positive spin Teukolsky variables (i.e. $\bflin$, $\fflin$, $\alin$) in Section \ref{sec:hierarchy-positive-spin}, and then for the negative spin variables (i.e. $\underline{\bflin}$, $\underline{\fflin}$, $\aalin$) in Section \ref{sec:hierarchy-negative-spin}.

Notice that the resulting estimates in the present section remain valid throughout the entire subextremal and extremal range $|Q| \leq M$. Consequently, an extension of Theorem \ref{thm:mastercontrol} to this full range would ensure that the derived hierarchy yields control over the Teukolsky variables in the same regime.

\subsection{The Hierarchy for Positive Spin Teukolsky Variables}\label{sec:hierarchy-positive-spin}

We define the following weighted energy fluxes associated to the Teukolsky variables $\bflin$, $\fflin$ and $\alin$:
\begin{align}\label{eq:definition-norms-positive-spin}
\begin{split}
        \mathrm{E}_{u}[\fflin]&\doteq \int_{v_0}^\infty \int_{S^2}\Big[|\Omega\ns_4(r^2\Omega \fflin)|^2+|[r\ns]r^2\Omega \fflin|^2+|r^2\Omega \fflin |^2\Big]dv\varepsilon_{S^2},\\
    \mathrm{E}_u[\bflin]&\doteq \int_{v_0}^\infty \int_{S^2}\Big[|\Omega\ns_4(r^4\Omega \bflin)|^2+|[r\ns]r^4\Omega \bflin|^2+|r^4\Omega \bflin|^2\Big]dv\varepsilon_{S^2},\\
    \mathrm{E}_{u}[\alin]&\doteq\int_{v_0}^{\infty}\int_{S^2}\Big[|\Omega\ns_4(r\Omega^2\alin)|^2+|[r\ns]r\Omega^2\alin|^2+|r\Omega^2\alin|^2\Big]dv\eps_{S^2},
\end{split}
\end{align}
and their higher derivative versions $\mathrm{E}^{i,j}_{u}[\fflin]$, $\mathrm{E}^{i,j}_u[\bflin]$, $\mathrm{E}^{i,j}_{u}[\alin]$ as above.

We obtain control in terms of the master energy and the initial data of $\mathrm{E}_{u}[\fflin]$, $\mathrm{E}_{u}[\bflin]$ and $\mathrm{E}_{u}[\alin]$ as follows.

\paragraph{Control of $\mathrm{E}_{u}[\mathfrak{f}]$:}

From the definition \eqref{eq:definitions-ff} of $\fflin\doteq \DDs_2 \bFlin +\rhoF \chilin$ and the fact that $\mathbb{E}^{i,j}[\mathcal{S}](u_1)$ contains $T$ and angular derivatives of $\chilin$ and $\bFlin$, we easily deduce that 
\begin{align*}
    \int_{v_0}^\infty \int_{S^2}|[\ns_T]^i[r\ns]^j r^2\Omega \fflin |^2 dv\varepsilon_{S^2} \lesssim \mathbb{E}^{i,j+1}[\mathcal{S}](u_1) \lesssim \mathbb{E}^{i,j+1}_{data}[\mathcal{S}](u),
\end{align*}
where for the last inequality we used Corollary \ref{corollary:higher-derivatives}.

Using \eqref{eq:definition-T} and equations \eqref{nabb-3-chih}-\eqref{nabb-3-bF}, we deduce the following equations in the $\ns_4$ direction:
\begin{align} \label{eq: e_4(chilin,bFlin)}
\begin{aligned}
\Omega\ns_4(r\Omega\xlin)&=2r\ns_T(\Omega\xlin)+2\Omega^2[r\slashed{\mathcal{D}}_2^{\star}]\Big(\elin+\frac{1}{2}\ns\lambda\Big)+\Omega^2(\Omega\xblin-\DDs_2[r\ns]\lambda]), \\
\Omega\ns_4(r\Omega\bFlin)&=2r\ns_T(\Omega\bFlin)+\Omega^2[r\slashed{\mathcal{D}}_1^{\star}](0,\sigmaFlin)-\Omega^2[r\ns]\Big(\rhoFlin-\frac{Q\lambda}{r^2}\Big)-2\frac{Q}{r}\Omega^2\Big(\elin+\frac{1}{2}\ns\lambda\Big),
\end{aligned}
\end{align}
where we manipulated the right hand side to have terms, such as $\rhoFlin-\frac{Q\lambda}{r^2}$, appearing in $\mathbb{E}^{i,j}[\mathcal{S}](u)$. For the other terms, observe that using the definition \eqref{eq:definition-lambda} of $\lambda$ and \eqref{eq:pr-Olin}, we can write 
\begin{align}\label{eq:expression-eta-nabla-lambda}
     \elin+   \frac{1}{2}\ns\lambda=-\eblin+\ns\frac{r\otx}{2\Omega^2},
    \end{align}
 which is controlled by   $\mathbb{E}^{i,j}[\mathcal{S}](u)$. For the remaining term, combining equation~\eqref{nabb-4-chibh} and \eqref{eq:partial-v-lambda} we obtain
     \begin{align}\label{eq:nab-4-chibh-lambda}
       \Omega\ns_4\left(r\Omega\xblin-[r\DDs_2][r\ns]\lambda\right)&=-2\Omega^2[r\DDs_2]\Big(\eblin-\frac{[r\ns]\otx}{2\Omega^2}\Big)-\Omega^2\Omega\xlin,
    \end{align}
and applying Lemma \ref{lem:tplemma} to the above this gives 
        \begin{align}\label{eq:control-intermediate-chibh-lambda}
        \begin{split}
        \int_{v_0}^{\infty}\int_{S^2}\frac{\Omega^2}{r^2}|[r\ns]^j[\ns_T]^i(r\Omega\xblin-[r\DDs_2][r\ns]\lambda)|^2dv\eps_{S^2}&\lesssim\|[\ns_T]^i[r\ns]^j(r\Omega\xblin-[r\DDs_2][r\ns]\lambda)\|_{u_1,v_0}^2\\
        &\qquad+\mathbb{E}^{i,j+1}[\mathcal{S}](u).
        \end{split}
    \end{align}
    By putting together the above, we deduce
    \begin{align*}
    \int_{v_0}^\infty \int_{S^2}| \Omega \ns_4[\ns_T]^i[r\ns]^jr^2\Omega \fflin |^2 dv\varepsilon_{S^2} &\lesssim \|[\ns_T]^i[r\ns]^j(r\Omega\xblin-[r\DDs_2][r\ns]\lambda)\|_{u_1,v_0}^2+\mathbb{E}^{i+1,j+1}[\mathcal{S}](u).
\end{align*}
Using Corollary \ref{corollary:higher-derivatives}, this proves that 
\begin{align}\label{eq:control-Eufflin}
    \mathrm{E}^{i,j}_{u}[\fflin] \lesssim \|[\ns_T]^i[r\ns]^j(r\Omega\xblin-[r\DDs_2][r\ns]\lambda)\|_{u_1,v_0}^2+\mathbb{E}^{i+1,j+2}_{data}[\mathcal{S}](u).
\end{align}

\paragraph{Control of $\mathrm{E}_{u}[\mathfrak{b}]$:}

From the definition \eqref{eq:definitions-bb} of $\bflin\doteq 2\rhoF \blin - 3 \rho \bFlin$, we are left to obtain control for $\blin$. By rewriting the Codazzi equation \eqref{Codazzi-chi} as
\begin{align*}
    r^2\Omega\blin&=-\Omega^2r \Big(\eblin-[r\ns]\frac{\otx}{2\Omega^2}\Big)-[r\divs](r\Omega\xlin)+r^2\rhoF\Omega\bFlin ,
\end{align*}
we deduce 
\begin{align}\label{eq:control-blin}
\int_{v_0}^{\infty}\int_{S^2}|[\ns_T]^i[r\ns]^jr^2\Omega\blin|^2dv\eps_{S^2}&\lesssim \mathbb{E}^{i,j+1}[\mathcal{S}](u_1) \lesssim \mathbb{E}^{i,j+1}_{data}[\mathcal{S}],
\end{align}
where for the last inequality we used Corollary \ref{corollary:higher-derivatives}.

{\color{black} Next, taking the $\ns_4$  derivative of the  Codazzi equation yields \begin{align*}
    \Omega\ns_4( r^2\Omega\blin)= &-\frac{\Omega^2}{r} \Omega\ns_4\Big(r^2\eblin -r^2[r\ns]\frac{\otx}{2\Omega^2}\Big) - \Omega\ns_4\left(\frac{\Omega^2}{r}\right)\Big(r^2\eblin -r^2[r\ns]\frac{\otx}{2\Omega^2}\Big) \\ &  -[r\divs]\Omega\ns_4(r\Omega\xlin)+\Omega\ns_4(r^2\rhoF\Omega\bFlin).
\end{align*}
However, using~\eqref{D-otx} and~\eqref{nabb-4-eblin}, it is immediate to check that \begin{align*}
    \Omega\ns_4\Big(r^2\eblin -r^2[r\ns]\frac{\otx}{2\Omega^2}\Big) = r^2\Omega\blin + r^2\rhoF \Omega\bFlin,
\end{align*}
and substituting $\Omega\ns_4(r\Omega\chilin), \Omega\ns_4(r\Omega\bFlin)$ with the expressions calculated in \eqref{eq: e_4(chilin,bFlin)}, we  readily obtain \begin{align*}
\int_{v_0}^{\infty}\int_{S^2}&\Big[|\Omega\ns_4([\ns_T]^i[r\ns]^jr^2\Omega\blin)|^2+| [r\ns][\ns_T]^i[r\ns]^jr^2\Omega\blin|^2+| [\ns_T]^i[r\ns]^jr^2\Omega\blin|^2\Big]dv\eps_{S^2}\\
&\lesssim \|[\ns_T]^i[r\ns]^{j+1}(r\Omega\xblin-[r\Ds_2][r\ns]\lambda)\|_{u_1,v_0}^2+\mathbb{E}^{i+1,j+2}[\mathcal{S}](u).
\end{align*}
Using Corollary \ref{corollary:higher-derivatives}, this proves that 
    \begin{align*}
\mathrm{E}^{i,j}_{u}[\bflin] &\lesssim \|[\ns_T]^i[r\ns]^{j+1}(r\Omega\xblin-[r\Ds_2][r\ns]\lambda)\|_{u_1,v_0}^2+\mathbb{E}_{data}^{i+1,j+2}[\mathcal{S}](u).
\end{align*}
}

\paragraph{Control of $\mathrm{E}_{u}[\alpha]$:}
{\color{black} In order to bound $[\ns_T]^i[r\ns]^jr\Omega^2\alin$ we work with equations \eqref{nabb-4-chih} and \eqref{eq: e_4(chilin,bFlin)} to write
\begin{align} \label{eq:alin alt-expr}
\begin{aligned}
         -r\Omega^2 \alin\  &=  \Omega\ns_4(r\Omega \chilin)-2r\Omega\hat{\om}\Omega\xlin+\frac{r}{2}\Omega\trch \Omega \chilin\\
      & = 2\ns_T (r\Omega\xlin)+2\Omega^2[r\slashed{\mathcal{D}}_2^{\star}]\Big(\elin+\frac{1}{2}\ns\lambda\Big)+\Omega^2(\Omega\xblin-\DDs_2[r\ns]\lambda]) +\left(\frac{1}{2}\Omega\trch -2\Omega\hat{\om}\right)(r\Omega  \chilin),
\end{aligned}
\end{align}
from which we deduce 
    \begin{align}\label{eq:control-alpha-1}
    \int_{v_0}^{\infty}\int_{S^2}\Big[|[\ns_T]^i[r\ns]^jr\Omega^2\alin|^2\Big]dv\eps_{S^2}&\lesssim \left\|[\ns_T]^i[r\ns]^j\left(r\Omega\xblin-[r\DDs_2][r\ns]\lambda\right)\right\|_{u_1,v_0}+\mathbb{E}^{i+1,j+1}[\mathcal{S}](u).
    \end{align} 
    It only remains to estimate the $\ns_4$ derivative of $\alin$, and working with \eqref{eq:alin alt-expr}, we calculate \begin{align*}
        - \Omega\ns_4 (r\Omega^2 \alin) &= 2\ns_T ( \Omega\ns_4(r\Omega\chilin)) + \frac{2\Omega^2}{r^2}[r\slashed{\mathcal{D}}_2^{\star}]\Omega\ns_4\Big(r^2\elin+\frac{r^2}{2}\ns\lambda\Big)+ 2 \Omega\ns_4 \left(\frac{\Omega^2}{r^2}\right)[r\slashed{\mathcal{D}}_2^{\star}]\Big(r^2\elin+\frac{r^2}{2}\ns\lambda\Big) \\ 
        &\qquad + \frac{\Omega^2}{r}\Omega\ns_4(r\Omega\xblin-r\DDs_2[r\ns]\lambda) + \Omega\ns_4\left(\frac{\Omega^2}{r}\right)(r\Omega\xblin-r\DDs_2[r\ns]\lambda) +\Omega\ns_4\Big[\Big(\frac{1}{2}\Omega\trch -r\Omega\hat{\om}\Big)(r\Omega  \chilin)\Big].
    \end{align*}
    Thus, using relations \eqref{eq: e_4(chilin,bFlin)}-\eqref{eq:nab-4-chibh-lambda} we arrive at \begin{align*}
    - \Omega\ns_4 (r\Omega^2 \alin) &= 4[\ns_T]^2(r\Omega\chilin)  +\frac{4\Omega^2}{r^2}[\ns_T][r\slashed{\mathcal{D}}_2^{\star}]\Big(r^2\elin+\frac{r^2}{2}\ns\lambda\Big)+\frac{2\Omega^2}{r}[\ns_T](r\Omega\xblin-r\DDs_2[r\ns]\lambda) \\
    & -\frac{2\Omega^2}{r^2} [r\Ds_2](r^2\Omega\blin +r^2\rhoF \Omega\bFlin) + 2 \Omega\ns_4 \left(\frac{\Omega^2}{r^2}\right)[r\slashed{\mathcal{D}}_2^{\star}]\Big(r^2\elin+\frac{r^2}{2}\ns\lambda\Big)-\frac{\Omega^4}{r^2}(r\Omega\xlin) \\
    &  -\frac{2\Omega^4}{r}[r\DDs_2]\Big(\eblin-\frac{[r\ns]\otx}{2\Omega^2}\Big) + \Omega\ns_4\left(\frac{\Omega^2}{r}\right)(r\Omega\xblin-r\DDs_2[r\ns]\lambda) +\Omega\ns_4\Big[\Big(\frac{1}{2}\Omega\trch -r\Omega\hat{\om}\Big)(r\Omega  \chilin)\Big]
    \end{align*} 
    from which we immediately see that \begin{align*}
\int_{v_0}^{\infty}\int_{S^2}|\Omega\ns_4(r\Omega^2\alin)|^2dv\eps_{S^2}
\lesssim \left\|[\ns_T]^{\leq 1}\left(r\Omega\xblin-[r\Ds_2][r\ns]\lambda\right)\right\|_{u_1,v_0}^2+\mathbb{E}^{2,2}[\mathcal{S}](u). 
    \end{align*}
    Using Corollary \ref{corollary:higher-derivatives} this proves that 
    \begin{align*}
\mathrm{E}^{i,j}_{u}[\alin] &\lesssim \left\|[\nab_T]^{i+1}[r\ns]^{j+1}\left(r\Omega\xblin-[r\Ds_2][r\ns]\lambda\right)\right\|_{u_1,v_0}^2+\mathbb{E}_{data}^{i+2,j+2}[\mathcal{S}](u).
\end{align*}
}

\subsection{The Hierarchy for Negative Spin Teukolsky Variables}\label{sec:hierarchy-negative-spin}

We define the following weighted energy fluxes associated to the Teukolsky variables $\underline{\bflin}$, $\underline{\fflin}$ and $\ablin$:
\begin{align}\label{eq:definition-norms-negative-spin}
\begin{split}
       \mathrm{E}_{u}[\underline{\fflin}]&\doteq \int_{v_0}^\infty \int_{S^2}\Big[|\Omega\ns_4(r^2\Omega \underline{\fflin})|^2+|[r\ns]r\Omega^2 \underline{\fflin}|^2+|r\Omega^2\underline{\fflin}|^2\Big]dv\varepsilon_{S^2}\\
    \mathrm{E}_u[\underline{\bflin}]&\doteq \int_{v_0}^\infty \int_{S^2}\Big[|\Omega\ns_4(r^4\Omega \underline{\bflin})|^2+|[r\ns]r^3\Omega^2\underline{\bflin}|^2+|r^3\Omega^2\underline{\bflin}|^2\Big]dv\varepsilon_{S^2}\\
    \mathrm{E}_{u}[\ablin]&\doteq\int_{v_0}^{\infty}\int_{S^2}\Big[|r\Omega\ns_4(r\Omega^2\ablin)|^2+|[r\ns]\Omega^3\ablin|^2+|\Omega^3\ablin|^2\Big]dv\eps_{S^2}, 
\end{split}
\end{align}
and their higher derivative versions $\mathrm{E}^{i,j}_{u}[\underline{\fflin}]$, $\mathrm{E}^{i,j}_u[\underline{\bflin}]$, $\mathrm{E}^{i,j}_{u}[\ablin]$ as above.

We obtain control in terms of the master energy and the initial data of $\mathrm{E}_{u}[\underline{\fflin}]$, $\mathrm{E}_{u}[\underline{\bflin}]$ and $\mathrm{E}_{u}[\ablin]$ as follows.

\paragraph{Control of $\mathrm{E}_{u}[\underline{\mathfrak{f}}]$:}

We can rewrite the definition \eqref{eq:definitions-ff} of $\underline{\fflin}\doteq \Ds_2\bbFlin-\rhoF\xblin$ as 
     \begin{align*}
        r^2\Omega\underline{\fflin}&=r[r\Ds_2]\Big(\Omega \bbFlin-\rhoF[r\ns]\lambda\Big)-r\rhoF\Big(r\Omega\xblin-[r\DDs_2][r\ns]\lambda\Big),
    \end{align*}
    where in the above we have the sum of two $f(v)$-gauge invariant quantities, according to Lemma \ref{lem:exactsol}. Now, from \eqref{nabb-4-bbF}-\eqref{eq:partial-v-lambda} and \eqref{eq:nab-4-chibh-lambda},  their $\ns_4$ derivatives are given by
    \begin{align}
        \Omega\ns_4\Big(r\Omega \bbFlin-r^2\rhoF \frac{1}{r}[r\ns]\lambda\Big)&=r\Omega^2\DDs_1(0, -\sigmaFlin)-\Omega^2[r\ns]\Big(\rhoFlin-\frac{Q\lambda}{r^2}\Big) \nonumber\\
        &\qquad-2\Omega^2\rhoF r\Big(\eblin- \frac{1}{2\Omega^2}[r\ns]\otx\Big), \label{eq:nab-4-bbbF-lambda}\\
                \Omega\ns_4\Big(r\rhoF[r\Omega\xblin-[r\DDs_2][r\ns]\lambda]\Big)&=-2\Omega^2r^2\rhoF[r\DDs_2]\Big(\eblin-\frac{[r\ns]\otx}{2\Omega^2}\Big)-\Omega^2\Omega r^2\rhoF\xlin \nonumber\\
        &\qquad -\frac{\Omega^2}{r^2}[r^2\rhoF]\big(r\Omega\xblin-[r\DDs_2][r\ns]\lambda\big),\nonumber
    \end{align}
    where we manipulated the right hand side to have terms, such as $\rhoFlin-\frac{Q\lambda}{r^2}$, appearing in $\mathbb{E}^{i,j}[\mathcal{S}](u)$, or terms, such as $r\Omega\xblin-[r\DDs_2][r\ns]\lambda$, which have already been estimated (see \eqref{eq:control-intermediate-chibh-lambda}).

    By applying Lemma~\ref{lem:tplemma}, we deduce
    \begin{align}\label{eq:control-bbF-lambda}
    \begin{split}
    &\int_{v_0}^{\infty}\int_{S^2}\frac{\Omega^2}{r^2}\Big[|r\big(\Omega \bbFlin-\rhoF[r\ns]\lambda\big)|^2+|r\rhoF\big(r\Omega\xblin-[r\DDs_2][r\ns]\lambda\big)|^2\Big]dv\eps_{S^2}\\
    &\lesssim \int_{v_0}^{\infty}\int_{S^2}\Big[\Big|\Omega\ns_4\Big(r(\Omega \bbFlin-\rhoF[r\ns]\lambda)\Big)\Big|^2+\Big|\Omega\ns_4\big(r\rhoF(r\Omega\xblin-[r\DDs_2][r\ns]\lambda)\big)\Big|^2\Big]dv\eps_{S^2}\\
    &\lesssim \left\|r(\Omega \bbFlin-\rhoF[r\ns]\lambda)\right\|_{u_1,v_0}^2+\left\|r\Omega\xblin-[r\DDs_2][r\ns]\lambda\right\|_{u_1,v_0}^2+\mathbb{E}^{0,1}[\mathcal{S}](u).
    \end{split}
\end{align}
Using Corollary \ref{corollary:higher-derivatives}, this proves that 
\begin{align}\label{eq:control-ffb}
\begin{split}
        \mathrm{E}^{i,j}_{u}[\underline{\fflin}] &\lesssim \|[\ns_T]^i[r\ns]^{j+1}\left(r(\Omega \bbFlin-\rhoF[r\ns]\lambda)\right)\|_{u_1,v_0}^2+\|[\ns_T]^i[r\ns]^{j+1}\left(r\Omega\xblin-[r\DDs_2][r\ns]\lambda\right)\|_{u_1,v_0}^2\\
    &\qquad+\mathbb{E}^{i,j+2}_{data}[\mathcal{S}](u).
\end{split}
\end{align}

    \paragraph{Control of $\mathrm{E}_{u}[\underline{\mathfrak{b}}]$:}

We can rewrite the definition \eqref{eq:definitions-bb} of $\underline{\bflin}\doteq 2\rhoF \bblin -3\rho \bbFlin$ as 
\begin{align*}
   r^4\Omega\underline{\bflin}=2r^2\rhoF \Big(r^2\Omega\bblin-\frac{3}{2}r^2\rho[r\ns]\lambda\Big)- 3r^3\rho\Big(r\Omega \bbFlin-r^2\rhoF \frac{1}{r}[r\ns]\lambda\Big),
\end{align*}
where we have written the above as the sum of two $f(v)$-gauge invariant quantities, according to Lemma \ref{lem:exactsol}. Now, from \eqref{eq:nab-4-bbbF-lambda} and \eqref{nabb-4-bb}-\eqref{eq:partial-v-lambda},  their $\ns_4$ derivatives are given by
\begin{align*}
    \Omega\ns_4\Big(r^3\rho\big(r\Omega \bbFlin-r^2\rhoF \frac{1}{r}[r\ns]\lambda\big)\Big)&=r^3\rho\Big[r\Omega^2\DDs_1(0, -\sigmaFlin)-\Omega^2[r\ns]\Big(\rhoFlin-\frac{Q\lambda}{r^2}\Big)\\
    &\qquad -2\Omega^2\rhoF r\Big(\eblin- \frac{1}{2\Omega^2}[r\ns]\otx\Big)\Big]-2r^2\rhoF^2\Omega^2\Big(r\Omega \bbFlin-r\rhoF [r\ns]\lambda\Big)\\
    \Omega\ns_4\Big(r^2\Omega\bblin-\frac{3}{2}r^2\rho[r\ns]\lambda\Big)&=r^2\Omega^2\DDs_1(0,\slin)-r\Omega^2[r\ns]\Big(\rlin-\big(\frac{3}{2}\rho+\rhoF^2\big)\lambda\Big)-3r^2\Omega^2\rho\Big(\eblin- \frac{1}{2\Omega^2}[r\ns]\otx\Big)\\
&\qquad-\rhoF r\Omega^2\Big(\Omega\bbFlin-\rhoF[r\ns]\lambda\Big)-r\rhoF\Omega^2[r\ns]\Big( \rhoFlin-\frac{Q\lambda}{r^2}\Big)\\
&\qquad+\rhoF r^2\Omega^2\Big(\DDs_1(0,-\sigmaFlin)+\tr\chi\bFlin\Big).
\end{align*}
    where we manipulated the right hand side to have terms, such as $\rhoFlin-\frac{Q\lambda}{r^2}$, appearing in $\mathbb{E}^{i,j}[\mathcal{S}](u)$, and terms, such as $r\big(\Omega \bbFlin-\rhoF[r\ns]\lambda\big)$, which have already been estimated (see \eqref{eq:control-bbF-lambda}), and additional terms given by $\slin$ and $\rlin-\big(\frac{3}{2}\rho+\rhoF^2\big)\lambda$.
    
To control those, we observe the following. From \eqref{eq:curl-elin-slin}, we deduce
\begin{align*} 
r^2|r\Omega\slin|^2&= r^2|\curls r\Omega\eblin|^2=r^2\Big|\curls r\Omega\Big(\eblin-\frac{[r\ns]\otx}{2\Omega^2}\Big)\Big|^2\lesssim \Big|[r\ns]r\Omega\Big(\eblin-\frac{[r\ns]\otx}{2\Omega^2}\Big)\Big|^2,
\end{align*} 
which yields
\begin{align*}
    \int_{v_0}^{\infty}\int_{S^2}\Big[|[\ns_T]^i[r\ns]^jr\Omega\slin|^2\Big]r^2dv\eps_{S^2}&\lesssim\mathbb{E}^{i,j+1}[\mathcal{S}](u).
\end{align*}
From \eqref{eq:nabb-4-rlin}-\eqref{eq:partial-v-lambda}, the $\ns_4$ derivative of $\rlin-\big(\frac{3}{2}\rho+\rhoF^2\big)\lambda$ is given by
\begin{align*}
        \Omega\ns_4\Big(r^3 \rlin - r^3(\frac{3}{2}\rho+\rhoF^2)\lambda \Big)&=\Omega\ns_4\Big(r^3 \rlin +3M\Big(1-\frac{r_c}{r}\Big)\lambda \Big)\\
        &=  \Omega[r\divs ] r^2\big(\blin+  \rhoF\bFlin\big)-4\Omega^2 Q \Big(\rhoFlin-\frac{Q\lambda}{r^2}\Big)
\end{align*}
  By applying Lemma~\ref{lem:tplemma} to the above, we deduce using \eqref{eq:control-blin}:
\begin{align*}
    \int_{v_0}^{\infty}\int_{S^2}\frac{\Omega^2}{r^2}\Big|r^3 \rlin - r^3(\frac{3}{2}\rho+\rhoF^2)\lambda\Big|^2dv\eps_{S^2}&\lesssim  \|r^3 \rlin - r^3(\frac{3}{2}\rho+\rhoF^2)\lambda\|_{u_1,v_0}^2+\mathbb{E}^{0,2}[\mathcal{S}](u).
\end{align*}
Finally, combining all the above and applying Lemma ~\ref{lem:tplemma}
we deduce using \eqref{eq:control-bbF-lambda}
    \begin{align}\label{eq:controlbbf-bb-lambda}
    \begin{split}
    &\int_{v_0}^{\infty}\int_{S^2}\frac{\Omega^2}{r^2}\Big[|r^3\rho\big(r\Omega \bbFlin-r^2\rhoF \frac{1}{r}[r\ns]\lambda\big)|^2+|r^2\Omega\bblin-\frac{3}{2}r^2\rho[r\ns]\lambda|^2\Big]dv\eps_{S^2}\\
        &\lesssim \int_{v_0}^{\infty}\int_{S^2}\Big[|\Omega\ns_4\Big(r^3\rho\big(r\Omega \bbFlin-r^2\rhoF \frac{1}{r}[r\ns]\lambda\big)\Big)|^2+|\Omega \ns_4\Big(r^2\Omega\bblin-\frac{3}{2}r^2\rho[r\ns]\lambda\Big)|^2\Big]dv\eps_{S^2}\\
    &\lesssim \|r^3\rho\big(r\Omega \bbFlin-r^2\rhoF \frac{1}{r}[r\ns]\lambda\big)\|_{u_1,v_0}^2+\|r^2\Omega\bblin-\frac{3}{2}r^2\rho[r\ns]\lambda\|_{u_1,v_0}^2\\
    &\qquad+ \|r\Omega^2[r\ns]\Big(\rlin-\big(\frac{3}{2}\rho+\rhoF^2\big)\lambda\Big)\|_{u_1,v_0}^2+\|\left(r\Omega\xblin-[r\DDs_2][r\ns]\lambda\right)\|_{u_1,v_0}^2+\mathbb{E}^{0,3}[\mathcal{S}](u).
    \end{split}
\end{align}
Using Corollary \ref{corollary:higher-derivatives}, this proves that 
\begin{align}\label{eq:control-bbf}
\begin{split}
        \mathrm{E}^{i,j}_{u}[\underline{\bflin}] &\lesssim \|[\ns_T]^i[r\ns]^{j}\left(r(\Omega \bbFlin-\rhoF[r\ns]\lambda)\right)\|_{u_1,v_0}^2+\|[\ns_T]^i[r\ns]^{j}\left(r^2\Omega\bblin-\frac{3}{2}r^2\rho[r\ns]\lambda\right)\|_{u_1,v_0}^2\\
    &\qquad+ \|[\ns_T]^i[r\ns]^{j+1}\left(r\Omega^2\Big(\rlin-\big(\frac{3}{2}\rho+\rhoF^2\big)\lambda\Big)\right)\|_{u_1,v_0}^2+\|[\ns_T]^i[r\ns]^{j}\left(r\Omega\xblin-[r\DDs_2][r\ns]\lambda\right)\|_{u_1,v_0}^2\\
    &\qquad+\mathbb{E}^{i,j+3}_{data}[\mathcal{S}](u).
\end{split}
\end{align}

\paragraph{Control of $\mathrm{E}_{u}[\underline{\alpha}]$:}

Rewriting the Bianchi identity \eqref{eq:Bianchi-aa-ff-bb} as 
\begin{align*}
     Q r\  \Omega\ns_4 \left( r\Omega^2 \ablin\right) = \left(-6M +\frac{7Q^2}{r} \right) \Omega (r\Omega^2\underline{\fflin}) +\Omega(r\Ds_2)( r^3 \Omega^2\underline{\bflin}),
\end{align*}
and noticing that we control the terms on the right hand side by \eqref{eq:control-ffb}, \eqref{eq:control-bbf}, we can once again use Lemma~\ref{lem:tplemma} and Corollary \ref{corollary:higher-derivatives} to prove that 
    \begin{align*}
\mathrm{E}^{i,j}_{u}[\underline{\alin}] &\lesssim  \|[\ns_T]^i[r\ns]^{j}\left(r(\Omega \bbFlin-\rhoF[r\ns]\lambda)\right)\|_{u_1,v_0}^2+\|[\ns_T]^i[r\ns]^{j}\left(r^2\Omega\bblin-\frac{3}{2}r^2\rho[r\ns]\lambda\right)\|_{u_1,v_0}^2\\
    &\qquad+ \|[\ns_T]^i[r\ns]^{j+1}\left(r\Omega^2\Big(\rlin-\big(\frac{3}{2}\rho+\rhoF^2\big)\lambda\Big)\right)\|_{u_1,v_0}^2+\|[\ns_T]^i[r\ns]^{j}\left(r\Omega\xblin-[r\DDs_2][r\ns]\lambda\right)\|_{u_1,v_0}^2\\
    &\qquad+\mathbb{E}^{i,j+3}_{data}[\mathcal{S}](u).
\end{align*}

\appendix
\section{A Remark on the Optimality of the Charge-to-Mass Ratio}\label{app:optimality}
Recall from the proof of Theorem~\ref{thm:mastercontrol} that we encounter a restriction on charge-to-mass from estimating the boundary term arising from $E_{u_1}$, which can be written as, 
\begin{align*}
    E_{u_1}[\mathcal{S}](v_0,v_1)&=\int_{v_0}^{v_1}\int_{S^2}\Big[|r\Omega\xlin|^2+2|r\Omega\bFlin|^2+2|r\Omega\sigmaFlin|^2+2\Big|r\Omega\eblin-[r\ns]\frac{r\otx}{2\Omega}\Big|^2\Big]dv\varepsilon_{S^2}\\
    &\qquad+E_{u_1}^{\mathrm{est}}[\mathcal{S}](v_0,v_1),
\end{align*}
where,
\begin{align}
    E_{u_1}^{\mathrm{est}}[\mathcal{S}](v_0,v_1)\doteq \int_{v_0}^{v_1}\int_{S^2}2|r\Omega[\rhoFlin-\rhoF\lambda]|^2dv\varepsilon_{S^2} +\int_{S^2}\Big[-\Big(1-\frac{r_c}{r}\Big)\frac{3M}{2}\lambda^2-\frac{1}{2}X\lambda\Big]\varepsilon_{S^2}\Big|_{(u_1,v_0)}^{(u_1,v_1)},\label{eqn:Eest}
\end{align}
are the terms involved in the estimation. 
Specifically the restriction on the charge-to-mass arises from the estimate through Lemma~\ref{lem:tplemma} from the equation, $\Omega\ns_4\xi=\Omega^2\Xi$ where, assuming $Q\neq 0$,
\begin{align*}
    \Xi=\frac{Q\lambda}{r^2}- \rhoFlin,\qquad \xi=\frac{1}{8Q}\Big(X+\Big(1-\frac{r_c}{r}\Big)6M\lambda\Big)
\end{align*}
Assume $F$ is a differentiable function such that $F\rightarrow a>0$ as $r\rightarrow \infty$. We compute that 
\begin{align*}
    \partial_v(F|\xi|^2)=r^2F'\Big|\frac{\Omega\xi}{r}\Big|^2+2F\Big\langle\Omega r\Xi,\frac{\Omega\xi}{r}\Big\rangle.
\end{align*}
Integrating from $(v_0,\infty)$ yields
\begin{align*}
a|\xi|^2(u_1,\infty)-\int_{v_0}^{\infty}\Big[r^2F'\Big|\frac{\Omega\xi}{r}\Big|^2+2F\Big\langle\Omega r\Xi,\frac{\Omega\xi}{r}\Big\rangle\Big]dv-F|\xi|^2(u_1,v_0)=0.
\end{align*}
Adding $0$ in this form to equation~\eqref{eqn:Eest} gives,
\begin{align}
    E_{u_1}^{\mathrm{est}}[\mathcal{S}](v_0,\infty)&= \int_{v_0}^{\infty}\int_{S^2}\Big[2|r\Omega\Xi|^2-r^2F'\Big|\frac{\Omega\xi}{r}\Big|^2-2F\Big\langle\Omega r\Xi,\frac{\Omega\xi}{r}\Big\rangle \Big]dv\varepsilon_{S^2}\label{eq:QuadForm}\\
    &\qquad +\int_{S^2}\Big[\frac{3M}{2}\Big[\frac{3Ma}{8Q^2}-1\Big]\lambda^2+2\times \frac{1}{4}\Big[\frac{3Ma}{8Q^2}-1\Big]X\lambda+\frac{a}{64Q^2}X^2\Big]\varepsilon_{S^2}\Big|_{(u_1,\infty)}\nonumber\\
    &\qquad\qquad+\int_{S^2}\Big[\Big(1-\frac{r_c}{r}\Big)\frac{3M}{2}\lambda^2+\frac{1}{2}X\lambda-F|\xi|^2\Big]\varepsilon_{S^2}\Big|_{(u_1,v_0)}\nonumber
\end{align}
Viewing the first two lines of the RHS of equation~\eqref{eq:QuadForm} as two  quadratic forms and appealing to Sylvester's criterion, the necessary and sufficient conditions for the RHS of equation~\eqref{eq:QuadForm} to be positive semi-definite are
\begin{enumerate}
    \item[a)]  $2r^2F'+F^2\leq 0,$
from which one can read off that the worst case scenario is
\begin{align*}
    F'+\frac{F^2}{2r^2}= 0\implies F(r)=\frac{a }{ 1-\frac{a}{2r}}.
\end{align*}
\item[b)]  $a\geq \frac{8Q^2}{3M}=2r_c.$
\end{enumerate}
This then gives that 
\begin{align*}
    F(r)=\frac{2r_c}{1-\frac{r_c}{r}}.
\end{align*}
We can then conclude from this discussion that for $F$ to be differentiable on the exterior, $r_+>r_c$, which is equivalent to our restriction on the charge-to-mass ratio.

{\small
\bibliographystyle{IEEEtran}
\bibliography{RNCL}}
\end{document}